\documentclass[11pt]{article}

\usepackage[margin=0.8in]{geometry}

\usepackage{booktabs} 
\usepackage[ruled]{algorithm2e} 

\SetAlFnt{\small}
\SetAlCapFnt{\small}
\SetAlCapNameFnt{\small}
\SetAlCapHSkip{0pt}
\IncMargin{-\parindent}

\usepackage{changepage}

\usepackage{natbib}
\usepackage{pifont}
\usepackage{authblk}
\usepackage{booktabs,multirow,array}
\usepackage[utf8]{inputenc} 
\usepackage[T1]{fontenc}    
\usepackage[colorlinks,citecolor=blue]{hyperref}       
\usepackage{url}            
\usepackage{booktabs}       
\usepackage{amsfonts}       
\usepackage{nicefrac}       
\usepackage{microtype}      
\usepackage{xcolor}         
\usepackage{amsmath}
\usepackage{amsthm,thmtools}
\usepackage{amssymb}
\usepackage{mathtools}
\usepackage{enumitem}

\usepackage{multicol}
\usepackage[normalem]{ulem}
\usepackage{cleveref}

\usepackage{soul,xcolor}

 \newcommand\numberthis{\addtocounter{equation}{1}\tag{\theequation}}

\newtheorem{claim}{Claim}
\newtheorem{theorem}{Theorem}
\newtheorem{lemma}{Lemma}
\newtheorem{proposition}{Proposition}
\newtheorem{observation}{Observation}
\newtheorem{definition}{Definition}

\newtheorem{corollary}{Corollary}

\newcommand{\calC}{\mathcal{C}}

\newcommand{\EFone}{\mathrm{EF}1}
\newcommand{\PROPone}{\mathrm{PROP}1}
\newcommand{\PROP}{\mathrm{PROP}}
\newcommand{\PO}{\mathrm{PO}}
\newcommand{\fPO}{\mathrm{fPO}}
\newcommand{\EF}{\mathrm{EF}}
\newcommand{\MBB}{\mathrm{MBB}}

\newcommand{\MRC}{\mathrm{MRC}}

\newcommand{\prices}{\mathbf{p}}
\newcommand{\rewards}{\mathbf{r}}
\newcommand{\earnings}{\mathbf{e}}
\newcommand{\budgets}{\mathbf{b}}

\newcommand{\dalloc}
{\mathbf{x}}
\newcommand{\dallocy}{\mathbf{y}}
\newcommand{\agents}{\mathcal{N}}
\newcommand{\items}{\mathcal{M}}

\newcommand{\mean}{w}
\newcommand{\bEF}{\beta\text{-}\mathrm{EF}}
\newcommand{\bPROP}{\beta\text{-}\mathrm{PROP}}

\newcommand{\ialloc}{\mathbf{a}}

\newcommand{\bEFk}{\beta\text{-}\mathrm{EF}k}
\newcommand{\bPROPk}{\beta\text{-}\mathrm{PROP}k}
\newcommand{\calI}{\mathcal{I}}

\begin{document}

\title{On the Fairness of Normalized $p$-Means \\ for Allocating Goods and Chores}
\date{}

\author[a]{Owen Eckart}
\author[a]{Alexandros Psomas}
\author[a]{Paritosh Verma}

\affil[a]{Department of Computer Science, Purdue University.

\texttt{{\{oeckart,apsomas,verma136\}@cs.purdue.edu}}}

\maketitle

\begin{abstract}
Allocating items in a fair and economically efficient manner is a central problem in fair division. We study this problem for agents with additive preferences, when items are all goods or all chores, divisible or indivisible. The celebrated notion of Nash welfare is known to produce fair and efficient allocations for both divisible and indivisible goods; there is no known analogue for dividing chores. The Nash welfare objective belongs to a large, parameterized family of objectives called the $p$-mean welfare functions, which includes other notable members, like social welfare and egalitarian welfare. However, among the members of this family, only the Nash welfare produces fair allocations for goods. Incidentally, Nash welfare is also the only member that satisfies the axiom of scale invariance, which is crucially associated with its fairness properties.

We define the class of ``normalized $p$-mean'' objectives, which imparts the missing key axiom of scale invariance to the $p$-mean family. Our results show that optimizing the normalized $p$-mean objectives produces fair and efficient allocations when the items are goods or chores, divisible or indivisible. For instance, the normalized $p$-means gives us an infinite class of objectives that produce $(i)$ proportional and Pareto efficient allocations for divisible goods, $(ii)$ approximately proportional and Pareto efficient allocations for divisible chores, $(iii)$ $\EFone$ and Pareto efficient allocations for indivisible goods for two agents, and $(iv)$ $\EFone$ and Pareto efficient allocations for indivisible chores for two agents.
\end{abstract}

\section{Introduction}

 We study a central problem in fair division: how to allocate items in a fair and economically efficient manner among $n$ agents with additive preferences. We focus on two settings: the case where items are all goods, i.e., every agent $i$ has a value $v_{ij} \geq 0$ for each item $j$, or they weakly prefer being allocated an item, and the case where all items are chores, i.e., every agent $i$ has a cost $c_{ij} \geq 0$ for each item $j$, or they weakly prefer not to be allocated an item. We further consider sub-cases based on whether the items (goods or chores) are all divisible or all indivisible.

Arguably the two most well-studied notions of fairness are \emph{proportionality} ($\PROP$) and \emph{envy-freeness} ($\EF$). Proportionality requires that each agent is at least as happy as in the allocation where she gets an $n^{th}$ of each item~\cite{steinhaus1948problem}.
Envy-freeness, a stronger notion than proportionality, requires that every agent prefers her allocation to the allocation of any other agent~\cite{foley1966resource}. For divisible items, envy-freeness, and thereby proportionality, can be simply achieved by splitting each item equally among all the agents. However, for indivisible items, even proportional allocations may not exist: consider the simple case of one indivisible item and two agents. Hence, discrete relaxations of envy-freeness and proportionality like \emph{envy-freeness up to one item} ($\EFone$) and \emph{proportionality up to one item} ($\PROPone$) are studied in the context of indivisible items. $\EFone$ and $\PROPone$ respectively require the $\EF$ and $\PROP$ guarantee to be satisfied upon the addition or removal of one item from an appropriate bundle~\cite{lipton2004approximately,conitzer2017fair}. Allocations satisfying $\EFone$ and $\PROPone$ always exist for agents with additive preferences and can be computed via simple procedures like round-robin, envy-cycle elimination, etc~\cite{lipton2004approximately,conitzer2017fair,bhaskar2021approximate,caragiannis2019unreasonable}. 
In terms of economic efficiency, the notion of Pareto efficiency is central. An allocation is Pareto efficient if increasing an agent's utility necessarily results in a decrease in another agent's utility. Pareto efficiency by itself is also easy to achieve, e.g. via a serial dictatorship. However, the problem of achieving fairness in conjunction with efficiency is significantly more challenging.


Half a century ago, Varian~\cite{varian1974equity} gave an elegant solution to this problem, for divisible items via the concept of \emph{competitive equilibrium of equal incomes} (CEEI). In CEEI, each agent is initially endowed with a budget of \$1. 
A competitive equilibrium is an allocation and prices for the items such that agents exhaust their budget, buy only their favorite bundle under the given prices, and the market clears. This competitive equilibrium is both economically efficient as well as envy-free. CEEI results in such allocations for both divisible goods and divisible chores. For divisible goods it is also known that allocations realized at CEEI are precisely the allocations that maximize the \emph{Nash welfare} objective~\cite{arrow1981handbook}, defined as the product of agents' utilities~\cite{nash1950bargaining}. Moreover, Nash welfare maximizing allocations for divisible goods (which are fair and efficient) can be computed by solving the convex program of Eisenberg and Gale~\cite{eisenberg1959consensus}.

The potent fairness properties of Nash welfare also translate to the case of indivisible goods: the Kalai-prize winning work of Caragiannis et al.~\cite{caragiannis2019unreasonable} shows that integral allocations that maximize the Nash welfare are $\EFone$, in addition to being Pareto efficient. Complementing this result, Yuen
and Suksompong~\cite{yuen2023extending} prove that Nash welfare is, in fact, the unique welfarist rule that produces $\EFone$ allocations. A welfarist rule for a \emph{welfare function} $f:\mathbb{R}^n \mapsto \mathbb{R}$ chooses an allocation $\ialloc$ that optimizes  $f(v_1(\ialloc_1), v_2(\ialloc_2), \ldots, v_n(\ialloc_n))$ among the set of all allocations, where $v_i(\ialloc_i)$ is the utility of agent $i$ in allocation $\ialloc$. Unfortunately, for indivisible chores there is no known analogue for Nash welfare~\cite{ebadian2022fairly,garg2023new}; no welfarist rule is known to even produce PROP1 and Pareto efficient allocations for indivisible chores.

This unreasonable fairness of the Nash welfare objective when allocating goods can be attributed to its strong axiomatic foundations: it satisfies numerous desirable axioms like symmetry, the irrelevance of unconcerned agents, the Pigou-Dalton Principle, and Pareto optimality, among others. Nash welfare belongs to a large family of objectives called the $p$-mean welfare functions, which also contain prominent objectives like \emph{social welfare} (sum of utilities of the agents), \emph{egalitarian welfare} (minimum utility among the agents), etc. Incidentally, all members of the $p$-mean family satisfy the aforementioned axioms. However, among members of this family, the Nash welfare alone is fair. This can be seen by considering an instance with two agents and ten goods, all of which are equally valued at $1$ by the first agent and at $1$ billion by the second: maximizing the social welfare results in agent 2 getting all the goods, maximizing the egalitarian welfare gives all but one item to agent 1, whereas maximizing the Nash welfare results in each agent getting five goods---a fair outcome.

A possible explanation for why Nash welfare stands out as the unique, remarkably fair member of the $p$-mean family is that the Nash welfare, alone, satisfies the axiom of scale invariance, which states that for any agent $i$, multiplying her valuations for all the items by a fixed positive number shouldn't change the outcome allocation. 
The significance of scale invariance is apparent in the previous example, where it plays a key role in explaining why the Nash welfare produces a fair allocation and others don't.


With an aim to rectify this inherent unfairness of the $p$-mean family, we define the family of ``normalized p-means” objectives, which, by design, makes the scale of values inconsequential. Simply put, the goal of this paper is to comprehensively explore the landscape of fairness and efficiency properties of allocations obtained by optimizing the family of normalized $p$-mean objectives. The entirety of our findings in this paper can be succinctly distilled in the following message: 

\begin{quote}
\textit{For agents with additive preferences, optimizing the normalized $p$-mean of utilities or disutilities results in fair and efficient allocations for divisible and indivisible goods, as well as divisible and indivisible chores.}
\end{quote}

\subsection{Our Contributions}


Let $\items$ be the set of items we'd like to allocate to $n$ additive agents.
Let us first define the class of normalized $p$-mean objectives. Recall that, for $p \in \mathbb{R}$, the generalized $p$-mean of numbers $s_1, \dots, s_n \geq 0$ is defined as $\mean_p(s_1, s_2, \ldots, s_n) \coloneqq \left( \frac{1}{n} \sum_{i \in [n]} s_i^p\right)^{\frac{1}{p}}$. As $p$ approaches $0$ this becomes the geometric mean, for $p=1$ it is the arithmetic mean, for $p=-1$ it is the harmonic mean, if $p$ approaches $-\infty$ it becomes the minimum function, and so on. The normalized $p$-mean of utilities of an allocation $\dalloc = (\dalloc_1, \dots, \dalloc_n)$ of goods is defined as $\mean_p\left( \frac{v_1(\dalloc_1)}{v_1(\items)}, \allowbreak \ldots \allowbreak, \frac{v_n(\dalloc_n)}{v_n(\items)} \right)$, where $v_i(\dallocy)$ is the value of agent $i$ for goods $\dallocy$. Similarly, for the case of chores, the normalized $p$-mean of disutilities is $\mean_p\left( \frac{c_1(\dalloc_1)}{c_1(\items)}, \ldots, \frac{c_n(\dalloc_n)}{c_n(\items)} \right)$, where $c_i(\dallocy)$ is the disutility of agent $i$ for getting chores $\dallocy$. That is, the normalized $p$-mean objective is simply the $p$-mean of the utilities/disutilities after they have been normalized so that the total utility/disutility for $\items$ is equal to $1$.
Note that these objectives do not fall under the definition of welfare functions (also known as collective utility functions), since the value of the normalized $p$-mean depends not only on the individual utilities/disutilities of the agents' but also on their value/cost for $\items$. 
A summary of results for optimizing normalized $p$-means of utilities/disutilities can be seen in~\Cref{table: results}.

\begin{table}[t]
\centering
\setlength{\extrarowheight}{4pt} 
\caption{Summary of our results for optimizing normalized $p$-means of utilities/disutilities.}
\begin{tabular}{lcc}
\toprule
 & \textbf{Goods} & \textbf{Chores} \\ \midrule
\multirow{4}{*}{\centering\textbf{Divisible}} & 
\multirow{3}{*}{%
    \begin{tabular}{@{}c@{}}
        \textcolor{green!70!black}{\ding{51}} PROP + PO for $p \leq 0$ \\[5pt]
        \cline{1-1}
        \textcolor{red}{\ding{56}} EF for $p \neq 0$ \\[5pt]
        \cline{1-1}
        \textcolor{red}{\ding{56}} PROP for $p > 0$ \\
    \end{tabular}
} & 
\multirow{3}{*}{%
    \begin{tabular}{@{}c@{}}
        \textcolor{green!70!black}{\ding{51}} $n^{1/p}$-PROP + PO for $p \geq 1$ \\[5pt]
        \cline{1-1}
        \textcolor{red}{\ding{56}} $n^{1/p}\left(1- \omega \left(\frac{\log{p}}{p}\right)\right)$-PROP for $p \geq 1$ \\[5pt]
        \cline{1-1}
        \textcolor{red}{\ding{56}} PROP for $p  \in \mathbb{R}$ \\
    \end{tabular}
} \\
& & \\ \\
& & \\ \midrule
\multirow{5}{*}{\centering\textbf{Indivisible}} & 
\multirow{4}{*}{%
    \begin{tabular}{@{}c@{}}
        \textcolor{green!70!black}{\ding{51}} PROP1 + PO via rounding for $p\leq 0$ \\[3pt]
        \cline{1-1}
        \textcolor{green!70!black}{\ding{51}} $n=2$: EF1 + PO for $p \leq 0$ \\[3pt]
        \cline{1-1}
        \textcolor{red}{\ding{56}} $n\geq 2$: PROP1 for $p > 0$ \\[3pt]
        \cline{1-1}
        \textcolor{red}{\ding{56}} $n>2$: PROP1 for $p < 0$ \\
    \end{tabular}
} & 
\multirow{4}{*}{%
    \begin{tabular}{@{}c@{}}
        \textcolor{green!70!black}{\ding{51}} $n^{1/p}$-PROP1 + PO via rounding for $p\geq 1$ \\[3pt]
        \cline{1-1}
        \textcolor{green!70!black}{\ding{51}} $n=2$: EF1 + PO for $p \geq 2$ \\[3pt]
        \cline{1-1}
        \textcolor{red}{\ding{56}} $n = 2$: PROP1 for $p < 2$ \\[3pt]
        \cline{1-1}
        \textcolor{red}{\ding{56}} $n > 2$: PROP1 for $p < n$  \\
    \end{tabular}
} \\
& & \\ 
& & \\ \\ \\ \bottomrule
\end{tabular}
\label{table: results}
\end{table}


\paragraph{Dividing Goods}
We start our investigation with the case of goods in~\Cref{SEC:GOODS}. For divisible goods, we show that \emph{any} allocation that maximizes the normalized $p$-mean of agents' utilities, for \emph{any} $p \leq 0$, is proportional, in addition to being fractionally Pareto efficient (\Cref{theorem:divisible-goods}). This holds for any number of agents with additive valuations and any number of goods. This complements the known result that Nash welfare (normalized $0$-mean) maximization leads to envy-free allocations, by giving an \emph{infinite family of objectives} that result in proportional allocations. This result is tight: for any $p > 0$, there exists an instance (even with two agents and two items) such that \emph{any} allocation that maximizes the normalized $p$-mean of utilities is not proportional (\Cref{theorem:div-goods-negative}). Additionally, the fairness guarantee of~\Cref{theorem:divisible-goods} cannot be strengthened to envy-freeness (\Cref{lemma:div-goods-not-EF}). Together, these results paint a complete picture of the fairness and efficiency properties of the normalized $p$-mean objective for divisible goods.

Moving to indivisible goods, we prove that we can, in fact, obtain $\PROPone$ and $\fPO$ integral allocations by rounding the divisible allocations obtained by maximizing the normalized $p$-mean of utilities for any $p \leq 0$ (\Cref{theorem:rounding-goods}). Towards proving this, we show that any divisible allocation that maximizes the normalized $p$-mean, for any $p \leq 0$, corresponds to a competitive equilibrium with \emph{unequal} incomes. We can then use a rounding algorithm of Barman and Krishnamurthy~\cite{barman2019proximity} and obtain a $\PROPone$ and $\fPO$ allocation.

Thus, by maximizing the normalized $p$-mean objective we can implicitly get $\PROPone$ allocations. However, can we achieve the stronger $\EFone$ guarantee? Surprisingly, we show that for two agents and any number of goods, any integral allocation that maximizes the normalized $p$-mean objective for \emph{any} $p \leq 0$ is $\EFone$ and $\PO$ (\Cref{theorem:indiv-goods-EF1}). Notably, this result seems to contradict the characterization of Yuen and Suksompong~\cite{yuen2023extending}---which holds even for two agents---that establishes Nash welfare as the only welfarist rule that produces $\EFone$ allocations. Our result is consistent with this characterization because, as mentioned earlier, the normalized $p$-mean objective is not a welfarist rule. Our result is tight in two orthogonal ways (\Cref{theorem:indiv-goods-tight-wrt-n}): first, maximizing the normalized $p$-mean objective for any $p>0$ and two agents doesn't yield $\PROPone$ (and hence not $\EFone$) allocations; second, for three (or more) agents, maximizing the normalized $p$-mean objective for $p<0$ also does not yield even $\PROPone$ allocations. 

\paragraph{Dividing Chores}
We study the case that items are chores in~\Cref{SEC:CHORES}. We begin by establishing sweeping negative results concerning the fairness properties of welfarist rules; see~\Cref{sec:chores negative result} for the definition. We prove that there is no welfarist rule for divisible chores that produces $\bEF$ allocations for \emph{any} $\beta \geq 1$, or $\bPROP$ allocations for $\beta \in [1,n)$ (\Cref{theorem:non-norm-neg-res}), noting that for chores, the $n$-$\mathrm{PROP}$ guarantee is trivially satisfied by every allocation. A similar impossibility can be established for indivisible chores (in fact, as a corollary to~\Cref{theorem:non-norm-neg-res}): there is no welfarist rule for indivisible chores that produces $\bEFk$ allocations for \emph{any} $\beta \geq 1$ and $k \geq 1$, or $\bPROPk$ allocations for \emph{any} $\beta \in [1,n)$ and $k \geq 1$, noting that $n$-$\mathrm{PROP}k$ is trivially satisfied by every allocation. See~\Cref{SEC:PRELIMS} for the definitions of these notions. These impossibility results show a stark contrast between chores and goods, where the Nash welfare rule produces $\EF$ (resp. $\EF1$) allocations for divisible (resp. indivisible) goods. 

Next, we prove that minimizing the normalized $p$-mean of agents' disutilities for $p\geq 1$ results in $n^{1/p}$-$\PROP$ and $\fPO$ allocations of divisible chores (\Cref{theorem:div-chore-prop}); this contrasts with the case of divisible goods, where we obtain an exact $\PROP$ guarantee. We show that this approximation guarantee is essentially tight. Specifically, for every $p \geq 1$, there exist instances with $n$ agents where every allocation that minimizes the normalized $p$-mean of disutilities is $n^{1/p}(1 - \Theta(\frac{\log{p}}{p}))$-$\PROP$ (\Cref{theorem:chores-div-neg-many-agents}). Additionally, minimizing the normalized $p$-mean objective for any $p \in \mathbb{R}$ does not result in $\PROP$ allocations, showing yet another interesting separation between divisible goods and chores (\Cref{theorem:chores-div-neg-many-agents}). 

For indivisible chores, allocations that satisfy $n^{1/p}$-$\PROPone$ and $\fPO$, for $p\geq 1$, can be obtained by first minimizing the normalized $p$-mean objective over divisible allocations, and then rounding using a procedure of~\cite{branzei2023algorithms,barman2019proximity} (\Cref{theorem:rounding-chores}). Similar to the case of goods, we again exploit a market interpretation of the normalized $p$-mean optimal divisible allocation.

Finally, our work identifies the first natural counterpart of Nash welfare for indivisible chores: for two agents and any number of indivisible chores, we prove that every integral allocation that minimizes the normalized $p$-mean of disutilities, for \emph{any} $p \geq 2$, is $\EFone$ and $\PO$ (\Cref{theorem:chore-EF1-result}). Interestingly, $p=2$ happens to be the transition point for the $\EFone$ guarantee, i.e.,  minimizing the normalized $(2-\epsilon)$-mean objective for any $\epsilon > 0$ does not result in $\EFone$ allocations. Specifically, for $n=2$ agents and all $p < 2$, minimizing the normalized $p$-mean of disutilities may not result in $\PROPone$ (and therefore not $\EFone$) allocations (\Cref{theorem:indiv-chores-negative-wrt-p}). Also, for three or more agents, and all $p < n$, minimizing the normalized $p$-mean of disutilities may not result in $\PROPone$ (and therefore not $\EFone$) allocations (\Cref{thm: 14}). Therefore, for all $p$, there exist instances for which optimizing the normalized $p$-mean of disutilities doesn't yield $\PROPone$ allocations. Interestingly, however, our results leave open the possibility that, given an instance with $n$ agents, a range of values of $p$ (with $p \geq n$) always works.

\subsection{Related Work}

Collective welfare, whether it is social, egalitarian, or Nash welfare, is one of the key objectives in computational economics and has been extensively studied in settings beyond fair division, including auctions~\cite{vickrey1961counterspeculation,clarke1971multipart}, voting~\cite{brandt2016handbook}, facility location~\cite{alon2010strategyproof}, job scheduling~\cite{graham1979optimization}, matching~\cite{abdulkadiroglu2013matching}, etc. 
The problem of (approximately) optimizing the $p$-mean welfare, while simultaneously providing fairness guarantees, has attracted significant attention in recent years. Barman et al.~\cite{barman2020tight} show, for indivisible goods and agents with subadditive valuations, how to efficiently compute an allocation that achieves an $O(n)$ approximation to the optimal $p$-mean welfare, for any $p \leq 1$, and, in the same setting, Chaudhury et al.~\cite{chaudhury2021fair} give a polynomial-time algorithm that outputs an allocation that satisfies a prominent fairness notion (namely, $1/2$-EFX or $(1 - \epsilon)$-EFX with bounded charity; see Chaudhury et al.~\cite{chaudhury2021fair} for the definitions of these notions) as well as achieves an $O(n)$ approximation to the $p$-mean welfare, for all $p \leq 1$ simultaneously. Barman and Sundaram~\cite{barman2021uniform} study indivisible goods and agents with identical subadditive valuations and show how to efficiently find an allocation that approximates the optimal $p$-mean welfare within a constant, for all $p \leq 1$. 
For indivisible goods and agents with identical and additive valuations, Garg et al.~\cite{garg2022tractable} give a PTAS for the problem of computing the $p$-mean welfare maximizing allocation. Viswanathan
and Zick~\cite{viswanathan2023general} study indivisible goods allocation under matroid rank valuations, and give a strategyproof mechanism which outputs the optimal \emph{weighted} $p$-mean welfare allocation for any $p \leq 1$, as well as guarantees several popular fairness criteria, e.g. the maximin fair share guarantee; in the same setting, maximizing $0$-mean welfare (i.e., Nash welfare) was shown to guarantee group strategyproofness~\cite{barman2022truthful}.
Finally, Barman et al.~\cite{barman2022universal} study divisible goods allocated online and show how to approximate the $p$-mean welfare for $p \leq 1$.
To the best of our knowledge, we are the first to consider $p$-mean welfare for the case of chores.

Valuation functions that are normalized to add up to $1$ are common in the computational social choice literature in the context of optimizing social welfare;  see Aziz~\cite{aziz2020justifications} for a discussion.
To the best of our knowledge, the only other work that normalizes valuations to explicitly maximize a function so that the output allocation satisfies additional (fairness) properties is Plaut and Roughgarden~\cite{plaut2020almost}, who prove that for two agents with additive valuations over indivisible goods, the leximin rule on normalized valuations gives EFX allocations (a notion stronger than $\EFone$).




While maximizing Nash welfare for indivisible goods results in $\EFone$ and $\PO$ allocations, the existence of such allocations for indivisible chores and additive costs remains an open problem; the absence of an analogue to Nash welfare for chores also makes this problem challenging. Despite this, $\EFone$ and $\PO$ allocations of indivisible chores are known to exist in various special cases like bivalued additive preferences~\cite{ebadian2022fairly, garg2022fair} and dichotomous supermodular preferences~\cite{barman2023fair}. Recently, Garg et al.~\cite{garg2023new} prove that for indivisible chores and agents with additive preferences, $\EFone$ and $\fPO$ allocations exist when there are three agents, as well as when there are at most two disutility functions; Garg et al.~\cite{garg2023new} also show that EFX and $\fPO$ allocations exist for three agents with bivalued disutilities. Branzei and Sandomirskiy~\cite{branzei2023algorithms} prove that $\PROPone$ and $\PO$ allocations exist for indivisible chores and additive agents, as well as give algorithms for finding such allocations. For divisible chores, while the existence of $\EF$ and $\PO$ allocations follow from the existence of CEEI~\cite{bogomolnaia2017competitive, varian1974equity}, their efficient computation remains an open problem~\cite{garg2022fair}.

\section{Preliminaries}\label{SEC:PRELIMS}


    We study the problem of allocating a set $\items$ of $m$ items among a set $\agents$ of $n$ agents. We consider two cases: the case that all items are goods, i.e., agents are (weakly) happier when receiving more items, and the case that all items are chores, i.e., agents are (weakly) less happy when receiving more items.
    We study both \emph{divisible} items and \emph{indivisible} items. Divisible items can be fractionally divided amongst the agents. A fractional allocation $\dalloc \in [0,1]^{n\cdot m}$ is a partition of $\items$ among $\agents$, where $x_{ij}$ denotes the fraction of item $j$ allocated to agent $i$, and $\dalloc_i = (x_{i1}, x_{i2}, \ldots, x_{im})$ represents all the items allocated to agent $i$. Indivisible items can be allocated to only one agent, i.e., $x_{ij} \in \{0,1\}$, and $\dalloc \in \{0,1\}^{n\cdot m}$ is referred to as an \emph{integral} allocation. 
    We often refer to $\dalloc_i$ as a \emph{bundle} of items.
    Whether our allocations are fractional or integral, we ask that all items are fully allocated, i.e. $\sum_{i \in \agents} x_{ij} = 1$ for all items $j \in \items$.
    For ease of exposition, we will use $\dalloc$ to denote fractional allocations and $\ialloc$ to denote integral allocations.
    For integral allocations $\ialloc$, it is often convenient to interpret $\ialloc_i$ as the subset of $\items$ allocated to agent $i$, therefore we overload notation and use $\ialloc_i \subseteq \items$ to also denote the set of items allocated to agent $i$. 

    \subsection{Preliminaries for Goods}

    When allocating goods, we assume that each agent $i \in \agents$ has a \emph{valuation (or utility) function} $v_i : [0,1]^m \mapsto \mathbb{R}_{\geq 0}$ which specifies the value $v_i(\dalloc_i)$ agent $i$ has for her allocation $\dalloc_i$. In this paper, we focus on \emph{additive} valuation functions. A valuation function $v_i$ is additive iff the value of any bundle of items $\dalloc_i$ can be expressed as $v_i(\dalloc_i) = \sum_{j=1}^n x_{ij} v_{ij}$, where $v_{ij} \in \mathbb{R}_{\geq 0}$ is the value that agent $i$ derives from being allocated good $j$ in its entirety. 


An (integral or fractional) allocation $\dalloc$ of goods is envy-free or $\EF$ iff for all pairs of agents $i, j \in \agents$, we have $v_i(\dalloc_i) \geq v_i(\dalloc_j)$, i.e., agent $i$'s value for her bundle is at least her value for $j$'s bundle. An (integral or fractional) allocation $\dalloc$ is \emph{proportional} or $\PROP$ iff for every agent $i \in \agents$, we have $v_i(\dalloc_i) \geq v_i(\items)/n$. Indeed, if an allocation is $\EF$ then it is $\PROP$. It is well known that $\EF$ and $\PROP$ allocations may not exist for indivisible goods: consider the case of two agents and one good that they both value. In this paper, we consider the following two well-studied relaxations of envy-freeness and proportionality for indivisible goods known as $\EFone$ and $\PROPone$. An integral allocation $\ialloc = (\ialloc_1, \ldots, \ialloc_n)$ is \emph{envy-free up to one good} or  $\EFone$ iff for all pairs of agents $i, j \in \agents$ with $\ialloc_j \neq \emptyset$ we have $v_i(\ialloc_i) \geq v_i(\ialloc_j \setminus \{g\})$, for some $g \in \ialloc_j$, i.e., agent $i$ prefers her bundle over agent $j$'s bundle after the removal of one item from agent $j$'s bundle. An integral allocation $\ialloc = (\ialloc_1, \ldots, \ialloc_n)$ is \emph{proportional up to one good} or $\PROPone$ if for every agent $i \in \agents$ there exists a good $g \in \items \setminus \ialloc_i$ such that $v_i(\ialloc_i \cup \{g\}) \geq v_i(\items)/n$.
    
    An (integral or fractional) allocation $\dalloc$ is \emph{fractionally Pareto efficient} or $\fPO$ iff there \emph{does not} exist any fractional allocation $\dallocy$ such that for all agents $i \in \agents$ we have $ v_i(\dallocy_i) \geq v_i(\dalloc_i)$,
    and there exists an agent $j \in \agents$ for which $ v_j(\dallocy_j) > v_j(\dalloc_j)$. Intuitively, in a Pareto optimal allocation, we cannot increase the utility of one agent without decreasing the utility of some other agent. If there does not exist any integral allocation $\dallocy$ satisfying the aforementioned conditions then $\dalloc$ is \emph{Pareto efficient} or $\PO$.

    \subsection{Preliminaries for Chores}

    When allocating chores, we assume that each agent $i \in \agents$ has a \emph{cost (or disutility) function} $c_i : [0,1]^m \rightarrow \mathbb{R}_{\geq 0}$. Similar to goods, we will assume that cost functions are additive, i.e., $c_i(\dalloc_i) = \sum_{j=1}^n x_{ij} c_{ij}$, where $c_{ij} \in \mathbb{R}_{\geq 0}$ is the cost or disutility that agent $i$ derives from being allocated chore $j$ in its entirety. 
    

    An (integral or fractional) allocation $\dalloc = (\dalloc_1, \ldots, \dalloc_n)$ of chores is $\beta$ \emph{approximately envy-free} or $\beta$-$\EF$ for $\beta \geq 1$ iff for all pairs of agents $i, j \in \agents$ we have $c_i(\dalloc_i) \leq \beta \cdot c_i(\dalloc_j)$. If $\beta = 1$, then $\dalloc$ is \emph{envy-free} or $\EF$. An (integral or fractional) allocation of chores $\dalloc$ is $\beta$-\emph{approximately proportional} or $\beta$-$\PROP$, for $\beta \geq 1$, iff for every agent $i\in \agents$ it holds that $c_i(\ialloc_i) \leq \beta \cdot c_i(\items)/n$. If $\beta = 1$, then $\dalloc$ is \emph{proportional} or $\PROP$. An integral allocation $\ialloc$ of chores is \emph{envy-free up to one chore} or $\EFone$ iff for all pairs of agents $i, j \in \agents$ with $\ialloc_i \neq \emptyset$, $c_i(\ialloc_i \setminus \{t\}) \leq c_i(\ialloc_j)$, for some $t \in \ialloc_i$.  An integral allocation $\ialloc$ is {\emph proportional up to one chore} or $\PROPone$ iff for every agent $i\in \agents$ with $\ialloc_i \neq \emptyset$ there exists a chore $t \in \ialloc_i$ such that $c_i(\ialloc_i \setminus \{t\}) \leq c_i(\items)/n$. 
    
    We also consider the following relaxations of $\EFone$ and $\PROPone$.
    \begin{definition}[$\bEFk$]
    An integral allocation $\ialloc$ of chores is $\bEFk$, for an integer $k \geq 1$ and $\beta \geq 1$, iff for all agents $i,j \in \agents$, we have $c_i(\ialloc_i \setminus S) \leq \beta \cdot c_i(\ialloc_j)$ for some $S \subseteq \ialloc_j$ with $|S| \leq k$.
    \end{definition} 

    \begin{definition}[$\bPROPk$]
    An integral allocation $\ialloc$ of chores is $\bPROPk$, for $\beta \geq 1$ and integer $k \geq 1$, if for every agent $i\in \agents$, there exists a set of chores $S \subseteq \ialloc_i$ with $|S| \leq k$ such that $c_i(\ialloc_i \setminus S) \leq \beta \cdot c_i(\items)/n$.
    \end{definition}

    For $\beta = 1$ and $k=1$, $\bEFk$ is exactly $\EFone$ and $\bPROPk$ is exactly $\PROPone$. Also, $\bEFk$ implies $\beta'\text{-}\mathrm{EF}k'$ and $\bPROPk$ implies $\beta'\text{-}\mathrm{PROP}k'$, for $\beta' \geq \beta$ and $k' \geq k$.
    
    An (integral or fractional) allocation $\dalloc$ of chores is $\fPO$ iff there does not exist a fractional allocation $\dallocy$ such that for all agents $i \in \agents$ we have $c_i(\dallocy_i) \leq c_i(\dalloc_i)$ and there is an agent $j \in \agents$ for which $c_j(\dallocy_j) < c_j(\dalloc_j)$. If there does not exist any integral allocation $\dallocy$ satisfying the aforementioned conditions then $\dalloc$ is \emph{Pareto efficient} or $\PO$.

\subsection{Divisible to Indivisible Reductions}

The following lemma establishes a relationship between the fairness properties obtained by optimizing any function of agents' (dis)utilities for indivisible and divisible items. Intuitively, if optimizing a function results in fair indivisible allocations, optimizing the same function will result in fair divisible allocations as well. This lemma will be useful in our technical sections when establishing positive results for divisible items, but also impossibility results for indivisible items.

\begin{lemma}\label{lemma:div-to-indiv}
    Let $f:\mathbb{R}^n \mapsto \mathbb{R}$ be a function such that minimizing $f$ for indivisible chore division instances results in $\bEFk$ (resp. $\bPROPk$) allocations. Then, minimizing $f$ for divisible chore division instances results in $\bEF$ (resp. $\bPROP$) allocations. Similarly, for the case of goods, if $f:\mathbb{R}^n \mapsto \mathbb{R}$ is a function such that maximizing $f$ for indivisible goods results in $\bEFk$ (resp. $\bPROPk$) allocations, then, maximizing $f$ for divisible goods results in $\bEF$ (resp. $\bPROP$) allocations.
\end{lemma}

\begin{proof}   
    We prove the lemma for the case of chores; the argument for the case of goods is identical. Consider any instance with a set $\items$ of divisible chores, and assume that minimizing $f$ in this instance results in an allocation $\dalloc^*$ that is not $\bEF$ (resp. $\bPROP$). This means that there exist agents $i, j$ such that $c_i(\dalloc^*_i) > \beta \cdot c_i(\dalloc^*_j)$ (resp., for $\PROP$, there exists an agent $i$ such that $c_i(\dalloc^*_i) > \beta \cdot c_i(\items)/n$). Let $\epsilon > 0$ be such that $c_i(\dalloc^*_i) - \epsilon = \beta \cdot c_i(\dalloc^*_j)$ (resp. $c_i(\dalloc^*_i) - \epsilon = \beta \cdot c_i(\items)/n$). 
    To arrive at a contradiction, we will construct an instance with indivisible chores such that (i) the allocation that corresponds to $\dalloc^*$ is feasible (and therefore optimal) with respect to $f$, and (ii) the cost of agent $i$ is strictly less than $\epsilon/k$ for all indivisible chores. We can achieve (i) as follows. Let $\hat{\mathcal{I}}$ be an instance with set of items $\hat{\items}$, where for each chore $t \in \items$, $\hat{\items}$ has $n$ chores, $t_1, \dots, t_n$, such that $t_p \in \hat{\items}$ corresponds to a $\dalloc^*_{pt}$ fraction of chore $t \in \items$. That is, for every agent $\ell$ the cost of the chore $t_p$, $\hat{c}_{\ell t_p} = \dalloc_{pt}^* \cdot c_{\ell t}$. Clearly, allocating chore $t_p$ to agent $p$ is optimal with respect to $f$, as it gives each agent $i$ disutility exactly $c_i(\dalloc^*_i)$. To additionally achieve (ii), consider an instance with a set of chores $\hat{\items^{\epsilon}}$, where $\hat{\items^{\epsilon}}$ is constructed by splitting every chore $t \in \hat{\items}$ into $z$ identical chores, for $z$ large enough such that agent $i$'s cost is at most $\epsilon/2k$ for every chore in $\hat{\items^{\epsilon}}$. It is easy to see that allocating all pieces that correspond to a chore $t \in \hat{\items}$ to the agent that got $t$ in instance $\hat{\mathcal{I}}$ gives each agent disutility $c_i(\dalloc^*_i)$ and is therefore optimal for $f$. However, this optimal allocation fails to satisfy the $\bEFk$ (resp. the $\bPROPk$) guarantee for agent $i$, because for any set of $k$ chores, $S$, we will have $c_i(\dalloc^*_i) - c_i(S) \geq c_i(\dalloc^*_i) -  \epsilon/2 > \beta \cdot c_i(\dalloc^*_j)$ (resp. $c_i(\dalloc^*_i) - c_i(S) \geq c_i(\dalloc^*_i) -  \epsilon/2 > \beta \cdot c_i(\items)/n$), contradicting the fact that $f$ results in $\bEFk$ (resp. $\bPROPk$) allocations.  
\end{proof}

\subsection{Preliminaries on Markets}

Given divisible goods $\items$ and agents $\agents$ having additive valuation functions $\{v_i\}_i$, a Fisher market equilibrium is defined as a triple $(\dalloc, \prices, \budgets)$ where $\dalloc$ is an allocation, $\prices = (p_1, p_2, \ldots, p_m) \in \mathbb{R}_{\geq0}^m$ defines the \emph{price} $p_j$ of each good $j\in \items$, and $\budgets = (b_1, b_2, \ldots, b_n) \in \mathbb{R}_{\geq0}^n$ denotes the \emph{budget} $b_i$ of each agent $i\in \agents$. A market equilibrium must satisfy the following three properties.
\begin{enumerate}
    \item \textbf{Market Clears:} For each good $j\in \items$, either $p_j=0$ or $\sum_{i=1}^n x_{ij} = 1$. 
    \item {\bf Maximum bang-per-buck allocation:} If $x_{ij}>0$, then $v_{ij}/p_j \geq v_{i\ell}/p_\ell$ for every item $\ell \in \items$.
    \item {\bf Budget exhausted.} For each agent $i$ we have $b_i = \sum_{j=1}^m p_j x_{ij}$. 
\end{enumerate}
In any market equilibrium $(\dalloc, \prices, \budgets)$ for each agent $i\in \agents$, we will denote the set of all maximum bang-per-buck goods by $\MBB_i \coloneqq \{j \in \items : v_{ij}/p_j \geq v_{i\ell}/p_\ell \text{ for all } \ell \in \items\}$. Additionally, overloading notations $\MBB_i$ will also be used to represents the maximum bang-per-buck value, $\MBB_i = \max_{j}{ v_{ij}/p_j}$.
\begin{proposition}[First Welfare Theorem;~\cite{mas1995microeconomic}, Chapter 16]\label{theorem:first-welfare}
    In any Fisher market equilibrium $(\dalloc, \prices, \budgets)$, the allocation $\dalloc$ is always fractionally Pareto efficient.
\end{proposition}

Fisher markets can analogously be defined for chores, see Appendix~\ref{appendix:market-prelims} for details.

\section{Allocating Goods}\label{SEC:GOODS}

In this section, we present our results for the case of goods. For ease of notation, for the remainder of this section, we use $\tilde{v}_i$ to refer to the \emph{normalized} valuation of agent $i$, i.e. $\tilde{v}_{ij} = v_{ij}/v_i(\items)$ and $\tilde{v}_{i}(\dalloc_i) = v_{i}(\dalloc_i)/v_i(\items)$.
For ease of exposition, we say that we maximize the \emph{normalized $p$-mean} of utilities to mean that we maximize the $p$-mean of normalized utilities, i.e., find the allocation $\dalloc$ with maximum $\mean_p( \tilde{v}_{i}(\dalloc_i), \ldots,  \tilde{v}_{n}(\dalloc_n) )$.

In~\Cref{sec:divisible goods}, we prove that for all $p \leq 0$, maximizing the normalized $p$-mean of utilities gives $\PROP$ and $\fPO$ allocations, and show that the KKT conditions of the corresponding convex program allow for a market interpretation of the optimal solution; for $p > 0$ this result does not hold. In~\Cref{sec:divisible goods rounding} we show how to round fractional $p$-mean maximizing allocations, for $p \leq 0$, to get integral allocations that are $\fPO$ and $\PROPone$. Finally, in~\Cref{sec:indivisible-goods} we prove that, for the case of two agents and $p \leq 0$, the integral allocation that maximizes the normalized $p$-mean of utilities is $\EFone$ and $\PO$; this is not true for more than two agents or $p > 0$. 

\subsection{Divisible Goods}\label{sec:divisible goods}

We prove the following theorem.

\begin{theorem}\label{theorem:divisible-goods}
    For any number of agents with additive valuations over divisible goods, and all $p \leq 0$, every allocation $\dalloc^*$ that maximizes the normalized $p$-mean of utilities is $\PROP$ and $\fPO$. Furthermore, there exist prices $p^*_j \geq 0$ for each item $j \in \items$ and budgets $b^*_i$ for each agent $i \in \agents$, such that $(\dalloc^*, \prices^*, \budgets^*)$ is a Fisher market equilibrium.
\end{theorem}
\begin{proof}   
Without loss of generality, we assume that for each good $j \in \items$, there exists an agent $i \in \agents$ who values it positively, i.e., $v_{ij} > 0$. Note that, items that have zero value for every agent can be assigned arbitrarily without affecting economic efficiency or fairness notions like $\EF$, $\PROP$, etc. 
We analyze the convex program whose objective is to maximize the normalized $p$-mean of agents' utilities for $p<0$ (the $p=0$ case follows from the equivalence of CEEI and Nash welfare optimization). For the sake of notational convenience, let $k \coloneqq -p$ where $k>0$. Noting that maximizing $\left( \sum_{i=1}^n \frac{1}{n} \left( \sum_{j=1}^m \tilde{v}_{ij}x_{ij} \right)^{-k}\right)^{-1/k}$ is equivalent to minimizing $\sum_{i=1}^n \left( \sum_{j=1}^m \tilde{v}_{ij}x_{ij} \right)^{-k}$, this convex program (indeed, the objective is convex, since each term of the objective $\left( \sum_{j=1}^m\tilde{v}_{ij}x_{ij} \right)^{-k}$ is convex) can be written as follows:
\begin{align*}
    \text{minimize} \ \ \sum_{i=1}^n & \left( \sum_{j=1}^m \tilde{v}_{ij}x_{ij} \right)^{-k}\\
    \text{subject to: }~~ \sum_{i=1}^n x_{ij} & \leq 1 \ \ \text{, for all } j \in \items\\
    -x_{ij} & \leq 0 \ \ \text{, for all } i \in \agents, j \in \items
\end{align*}
    \paragraph{Analysis of the Convex Program.} Let $\dalloc^*$ be the optimal solution to this program, and let $p^*_j \geq 0$ and $\kappa^*_{ij} \geq 0$ be the associated optimal dual variables corresponding to the constraints $\sum_{i=1}^n x_{ij} \leq 1$ and $-x_{ij} \leq 0$, respectively. The KKT conditions imply that

    \begin{align*}
        p^*_{j} \left( \sum_{i=1}^n x^*_{ij} - 1 \right) & = 0, \text{ for all }  j \in \items \label{equation:kkt-1} \numberthis\\
        \kappa^*_{ij}x^*_{ij} & = 0, \text{ for all } i \in \agents, j \in \items \label{equation:kkt-2} \numberthis\\
         \frac{\partial}{\partial x_{ij}}  \sum_{i=1}^n  \left( \sum_{j=1}^m \tilde{v}_{ij}x_{ij} \right)^{-k} \Biggr|_{x^*_{ij}} & + p^*_j \,  \frac{\partial}{\partial x_{ij}} \left(\sum_{i=1}^n x_{ij} - 1\right) \Biggr|_{x^*_{ij}} + \kappa^*_{ij} \,  \frac{\partial}{\partial x_{ij}} (-x_{ij}) \Biggr|_{x^*_{ij}} = 0 ,\forall i \in \agents, j \in \items. \numberthis
    \end{align*}

Note that, letting $\tilde{v}_i(\dalloc^*_i) = \sum_{j=1}^m \tilde{v}_{ij}x^*_{ij}$, stationarity (the third condition) implies:
\begin{equation}
k \tilde{v}_{ij} \, \tilde{v}_i(\dalloc^*_i)^{-(k+1)} = p^*_j - \kappa^*_{ij} \text{ for all } i \in \agents, j \in \items. \label{equation:kkt-3} 
\end{equation}

    If $\tilde{v}_{ij}>0$, then~\Cref{equation:kkt-3} implies that $k \, \tilde{v}_i(\dalloc^*_i)^{-(k+1)} = \frac{p^*_j - \kappa^*_{ij}}{\tilde{v}_{ij}}$. Note that, if $x^*_{ij} > 0$, then it must be that $\tilde{v}_{ij}>0$; otherwise, transferring item $j$ to an agent who values it positively increases the objective. Additionally, if $x^*_{ij} > 0$,~\Cref{equation:kkt-2} implies that $\kappa^*_{ij} = 0$, and therefore, $k \, \tilde{v}_i(\dalloc^*_i)^{-(k+1)} = \frac{p^*_j}{\tilde{v}_{ij}}$. Hence, if $x^*_{ij} > 0$, for any $\ell \in \items$ with $\tilde{v}_{i\ell}> 0$ we have
    \begin{equation}
        \frac{p^*_j}{\tilde{v}_{ij}} = k \, \tilde{v}_i(\dalloc^*_i)^{-(k+1)} = \frac{p^*_{\ell} - \kappa^*_{i\ell}}{\tilde{v}_{i\ell}} \leq \frac{p^*_{\ell}}{\tilde{v}_{i\ell}}. \label{equation:kkt-4}
    \end{equation}

    \paragraph{Market Interpretation.} 
    Let $b^*_i \coloneqq k \, \tilde{v}_i(\dalloc_i^*)^{-k}$. We prove that $(\dalloc^*, \prices^*, \budgets^*)$ is a market equilibrium.

    We begin by showing that $p^*_j > 0$ for all $j \in \items$. Observe that, for every agent $i \in \agents$, $\tilde{v}_i(\dalloc^*_i)>0$, since, as $\tilde{v}_i(\dalloc^*_i)$ approaches $0$ for some agent $i$, the objective tends to infinity; however, we can get a finite objective by splitting each good equally among the agents. Further, for each good $j \in \items$ we know that there's an agent $i$ such that $\tilde{v}_{ij}>0$. For such a pair $i,j$, since $\tilde{v}_i(\dalloc^*_i)>0$, $\tilde{v}_{ij}>0$, and $\kappa^*_{ij}\geq 0$,~\Cref{equation:kkt-3} implies that $p^*_j = k \tilde{v}_{ij} \, \tilde{v}_i(\dalloc^*_i)^{-(k+1)} + \kappa^*_{ij} > 0$.

    Towards proving that $(\dalloc^*, \prices^*, \budgets^*)$ is a market equilibrium, first note that the \emph{market clears} (i.e., $\sum_{i=1}^n x^*_{ij} = 1$, for all $j \in \items$): this is a direct implication of~\Cref{equation:kkt-1} and the fact that $p^*_j > 0$ for all $j \in \items$. Second, $\dalloc^*$ is a \emph{maximum bang-per-buck} allocation, i.e. if $x^*_{ij}>0$, then~\Cref{equation:kkt-4} implies that $\MBB_i \coloneqq  \frac{\tilde{v}_{ij}}{p^*_j} \geq \frac{\tilde{v}_{i\ell}}{p^*_\ell}$ for all $\ell \in \items$. The last inequality is trivially satisfied if $\tilde{v}_{i\ell}=0$; otherwise if $\tilde{v}_{i\ell}>0$ then it follows from~\Cref{equation:kkt-4}. Finally, every agent $i \in \agents$ exhausts their budget $b^*_i$. The amount spent by agent $i$ is
    \begin{align*}
        \sum_{j:\ x^*_{ij}>0} x^*_{ij} p^*_j &= \sum_{j:\ x^*_{ij}>0} x^*_{ij} \tilde{v}_{ij} \cdot \left( \frac{p^*_j}{\tilde{v}_{ij}} \right) \\
        &= \sum_{j:\ x^*_{ij}>0} x^*_{ij} \tilde{v}_{ij} \cdot \left( k \, \tilde{v}_i(\dalloc^*_i)^{-(k+1)} \right) \tag{\Cref{equation:kkt-4}}\\
        &= k \, \tilde{v}_i(\dalloc^*_i)^{-(k+1)} \sum_{j:\ x^*_{ij}>0} x^*_{ij} \tilde{v}_{ij} \\
        &=  k \, \tilde{v}_i(\dalloc^*_i)^{-(k+1)} \tilde{v}_i(\dalloc^*_i) \\
        &= b^*_i.
    \end{align*}

    \paragraph{Fairness and Efficiency.}
    The fact that $(\dalloc^*, \prices^*, \budgets^*)$ forms a market equilibrium and~\Cref{theorem:first-welfare} imply that $\dalloc^*$ is $\fPO$ with respect to $\tilde{v}_i(.)$s, and therefore with respect to $v_i(.)$s as well. To show that $\dalloc^*$ is $\PROP$ for agent $i$, we begin by considering $\tilde{v}_i(\dalloc^*_r)$ for an arbitrary agent $r \in \agents$:

    \begin{align*}
        \tilde{v}_i(\dalloc^*_r) = \sum_{\ell:\ x^*_{r\ell}>0} x^*_{r\ell} \cdot \tilde{v}_{i\ell} &  = \sum_{\ell :\ x^*_{r\ell}>0} x^*_{r\ell} \, p^*_\ell \left( \frac{\tilde{v}_{i\ell}}{p^*_\ell}\right) \tag{$p^*_j>0, \forall j \in \items$}\\
        & \leq \sum_{\ell :\ x^*_{r\ell}>0} x^*_{r\ell} \, p^*_\ell \cdot \frac{\tilde{v}_i(\dalloc^*_i)^{k+1}}{k} \tag{$\frac{\tilde{v}_{i\ell}}{p^*_\ell} \leq 
        \MBB_i =^{(Eq. \ref{equation:kkt-4})}  \frac{\tilde{v}_i(\dalloc^*_i)^{k+1}}{k}$}\\
        & = b^*_r \cdot \frac{\tilde{v}_i(\dalloc^*_i)^{k+1}}{k} \\
        & = \tilde{v}_r(\dalloc^*_r)^{-k} \cdot \tilde{v}_i(\dalloc^*_i)^{k+1}.
    \end{align*}
    Summing up the above inequality for all $r \in \agents$, and using the fact that the normalized utilities sum up to one, i.e., $\sum_{r=1}^n \tilde{v}_i(\dalloc^*_r) = 1$, we have $1 = \sum_{r=1}^n \tilde{v}_i(\dalloc^*_r) \leq   \tilde{v}_i(\dalloc^*_i)^{k+1} \sum_{r=1}^n \tilde{v}_r(\dalloc^*_r)^{-k}$, or
    \[
    \tilde{v}_i(\dalloc^*_i)^{k+1} \geq \left( \sum_{r=1}^n \tilde{v}_r(\dalloc^*_r)^{-k} \right)^{-1}.
    \]
    Recalling that $\sum_{r=1}^n \tilde{v}_r(\dalloc^*_r)^{-k}$ is the optimal value of our minimization objective (in our primal convex program), and that splitting each item equally among all agents is a feasible solution with value equal to $\sum_{r=1}^n \left( \frac{1}{n} \right)^{-k} = n^{k+1}$, we have
    \[
    \tilde{v}_i(\dalloc^*_i)^{k+1} \geq \left( \sum_{r=1}^n \tilde{v}_r(\dalloc^*_r)^{-k} \right)^{-1} \geq \left( n^{k+1} \right)^{-1}.
    \]
    Since, $k+1 > 0$, we get that $\tilde{v}_i(\dalloc^*_i) \geq 1/n$, or equivalently, $v_i(\dalloc^*_i) \geq v_i(\items)/n$. This establishes that $\dalloc^*$ is $\PROP$ with respect to $\tilde{v}_i(.)$s, which also implies that $\dalloc^*$ is $\PROP$ with respect to $v_i(.)$s.   
\end{proof}

Next, we show that~\Cref{theorem:divisible-goods} does not hold for $p >0$.

\begin{theorem}\label{theorem:div-goods-negative}
    For $n \geq 2$ agents with additive valuations over divisible goods, and all $p > 0$, there exist instances such that every allocation that maximizes the normalized $p$-mean is not $\PROP$.
\end{theorem}
\begin{proof}   
    We construct instances for $n=2$ agents. These constructions can be easily extended to any number of agents $n$, by including for each additional agent $i$ an item $i^*$, such that, $i$ only likes $i^*$ and $i^*$ is only liked by $i$.
    We consider the following three cases based on the value of $p$.

    \paragraph{Case 1. $p>1$.} Let $v_{ij} = 1/m$ for all agents $i \in \{ 1, 2 \}$ and goods $j$, where $m$ is the number of goods. Consider a $p$-mean maximizing allocation $\dalloc$, and let $\alpha \in [0,m]$ be the total fraction of goods that agent $1$ receives in $\dalloc$; $m-\alpha$ is the total fraction of goods agent $2$ receives. Given this, the $p$-mean of agents' normalized utilities can be written as $\left(\frac{1}{2} \left( \left( \frac{\alpha}{m} \right)^p + \left( \frac{m-\alpha}{m}\right)^p \right) \right)^{1/p}$. Since, $p>1$, the function $\left( \frac{\alpha}{m} \right)^p + \left( \frac{m-\alpha}{m}\right)^p$ is convex and is maximized when $\alpha=0$ or $\alpha=m$. Thus, one of the agents does not get any goods, and hence no $p$-mean optimal allocation satisfies $\PROP$.
    
    \paragraph{Case 2. $p = 1$.} Consider an instance with two goods such that $v_{11} > v_{21} > 1/2$ and $v_{21} < v_{22} < 1/2$. The unique allocation that maximizes the normalized $p=1$-mean of utilities allocates good $1$ to agent $1$ and good $2$ to agent $2$. This allocation is not $\PROP$ since agent $2$'s utility is less than $1/2$.
    
    \paragraph{Case 3. $0 < p < 1$.} Consider an instance with two goods such that $v_{11} = 1/2 + \varepsilon$ and $v_{12} = 1/2 - \varepsilon$ and $v_{21} = 1/2 + \delta$ and $v_{22} = 1/2 - \delta$. We will show that, for all $0<p<1$ and all $1/2 \geq \delta > \varepsilon > 0$, no $p$-mean maximizing allocation is $\PROP$. 

    First, in any optimal allocation, good $2$ must be allocated fully to agent $1$; we prove this in~\Cref{lemma:div_goods_counterexample}, which is stated in Appendix~\ref{app:missing proofs goods}. Now, assume that agent $1$ receives $\alpha^*$ fraction of item $1$ and agent $2$ receives a $1-\alpha^*$ fraction of item $1$ in some normalized $p$-mean maximizing allocation $\dalloc = (\dalloc_1, \dalloc_2)$. The agents' (normalized) utilities then are $\tilde{v}_1(\dalloc_1) = (1/2+\varepsilon)\alpha^*+ 1/2 - \varepsilon$ and $\tilde{v}_2(\dalloc_2) = (1/2 + \delta)(1-\alpha^*)$. We first compute the value of $\alpha^*$, and then show that agent $1$ doesn't get her fair-share of value, i.e., $\tilde{v}_1(\dalloc_1) < 1/2$, proving that $\dalloc$ is not $\PROP$.
    
    Since $\dalloc$ is $p$-mean maximizing, the value of $f(x) = \left( (1/2+\varepsilon)x+ 1/2 - \varepsilon \right)^p + \left((1/2 + \delta)(1-x)\right)^p$ is maximized (in the interval $[0,1]$) at $x = \alpha^*$. To compute this maximum we can differentiate $f(x)$ and equate it to zero: 
    \begin{align*}
        & \frac{df}{dx} \Biggr|_{\beta^*} = p\left( \frac{1}{2}+\varepsilon \right)\left(  \left(\frac{1}{2}+\varepsilon \right)\beta^*+ \frac{1}{2} - \varepsilon \right)^{p-1} - p\left(\frac{1}{2} + \delta \right)^p(1-\beta^*)^{p-1} = 0 \\
         \implies & \left( \frac{1}{2}+\varepsilon \right)^{\frac{1}{p-1}} \left(  \left(\frac{1}{2}+\varepsilon \right)\beta^*+ \frac{1}{2} - \varepsilon \right) = \left(\frac{1}{2} + \delta \right)^{\frac{p}{p-1}} (1-\beta^*).
    \end{align*}
    Re-arranging we have:
    \begin{align}\label{equation:val-x-opt}
         \beta^* = \frac{\left(\frac{1}{2} + \delta \right)^{\frac{p}{p-1}} - \left( \frac{1}{2} - \varepsilon \right) \left( \frac{1}{2}+\varepsilon \right)^{\frac{1}{p-1}} }{\left(\frac{1}{2} + \delta \right)^{\frac{p}{p-1}} + \left( \frac{1}{2}+\varepsilon \right)^{\frac{p}{p-1}}}.
    \end{align}
    First, it is easy to verify that the value of $\beta^*$ above indeed corresponds to the maximum value of $f(x)$, since $\frac{d}{dx}f(\widehat{x})$ is positive for all $\widehat{x} < \beta^*$ and negative for all $\widehat{x} > \beta^*$. Second, $\beta^* \leq 1$ since the denominator of $\beta^*$ is at least the numerator. Therefore, the maximum in the interval $[0,1]$ occurs at $\alpha^* = \max\{\beta^*,0\}$.

    Now towards completing the proof, we will show that $\tilde{v}_1(\dalloc_1) = (1/2+\varepsilon)\alpha^*+ 1/2 - \varepsilon < \frac{1}{2}$, i.e., $\dalloc$ is not $\PROP$. This is equivalent to showing that $\alpha^* = \max\{\beta^*,0\} < \frac{\varepsilon}{1/2 + \varepsilon}$. Since $\frac{\varepsilon}{1/2+\varepsilon}>0$, this in turn, is equivalent to showing that $\beta^* <\frac{\varepsilon}{1/2 + \varepsilon}$, which can be established as follows:
    \begin{align*} 
        \beta^* - \frac{\varepsilon}{1/2 + \varepsilon} &=  \frac{\left(\frac{1}{2} + \delta \right)^{\frac{p}{p-1}} - \left( \frac{1}{2} - \varepsilon \right) \left( \frac{1}{2}+\varepsilon \right)^{\frac{1}{p-1}} }{\left(\frac{1}{2} + \delta \right)^{\frac{p}{p-1}} +  \left( \frac{1}{2}+\varepsilon \right)^{\frac{p}{p-1}}} - \frac{\varepsilon}{1/2 + \varepsilon} \tag{via (\ref{equation:val-x-opt})}\\
        &= \frac{\frac{1}{2}\left( \frac{1}{2}+\delta \right)^{\frac{p}{p-1}} - \frac{1}{2}\left( \frac{1}{2}+\varepsilon \right)^{\frac{p}{p-1}}}{\left( \left(\frac{1}{2} + \delta \right)^{\frac{p}{p-1}} +  \left( \frac{1}{2}+\varepsilon \right)^{\frac{p}{p-1}} \right)\left( \frac{1}{2}  + \varepsilon \right)}. \numberthis \label{equation:goods-no-prop}
    \end{align*}
    Finally, note that, $\delta > \varepsilon$ and $\frac{p}{p-1} < 0$ (since $0<p<1$). Therefore, we have $\frac{1}{2}\left( \frac{1}{2}+\delta \right)^{\frac{p}{p-1}} < \frac{1}{2}\left( \frac{1}{2}+\varepsilon \right)^{\frac{p}{p-1}}$. Along with~\Cref{equation:goods-no-prop}, this gives us the desired inequality, $\beta^* < \frac{\varepsilon}{1/2 + \varepsilon}$. This completes the proof.  
\end{proof}

Complementing~\Cref{theorem:divisible-goods}, the following lemma shows that maximizing the normalized $p$-mean of utilities does not produce $\EF$ allocations; its proof appears in Appendix~\ref{app:missing proofs goods}.

\begin{lemma}\label{lemma:div-goods-not-EF}
    There exist instances with $n \geq 3$ agents where maximizing the normalized $p$-mean of utilities for some $p<0$ doesn't produce $\EF$ allocations.
\end{lemma}


\subsection{Indivisible Goods: Rounding Divisible Allocations}\label{sec:divisible goods rounding}

Here, we prove that the market interpretation of $p$-mean optimal solutions for divisible goods in~\Cref{theorem:divisible-goods} can be leveraged to give new $\fPO$ and $\PROPone$ algorithms for indivisible goods. 

\begin{theorem}\label{theorem:rounding-goods}
For $n \geq 2$ agents with additive valuations over indivisible goods, and all $p \leq 0$, \Cref{algo:goods} outputs an $\fPO$ and $\PROPone$ allocation.
\end{theorem}

Our proof crucially uses the following result of Barman and Krishnamurthy~\cite{barman2019proximity}:

\begin{theorem}[Theorem 4; Barman and Krishnamurthy~\cite{barman2019proximity}]\label{theorem:rounding}
    Given a Fisher market equilibrium $(\dalloc,\prices, \budgets)$, we can, in polynomial time, find an \emph{integral allocation} $\ialloc$ that satisfies the following three properties:
    \begin{enumerate}
        \item $(\ialloc, \prices, \budgets')$ is a market equilibrium, $\budgets'$ satisfies $\sum_{i \in \agents} b'_i = \sum_{i \in \agents} b_i$ and $|b_i - b'_i| \leq ||\prices||_\infty$.
        \item If $b'_i < b_i$, then there exists an item $j \in \MBB_i$ such that $j \notin \ialloc_i$ and $b_i \leq b'_i + p_j$.
        \item If $b'_i > b_i$, then there exists an item $j \in \ialloc_i \subseteq \MBB_i$ such that $b_i \geq b'_i - p_j$.
    \end{enumerate}
\end{theorem}


\begin{algorithm}[t]
\textbf{Input:} {Instance with $m$ indivisible goods, $n$ agents with valuations $\{v_i\}_i$, and parameter $p \leq 0$.}
\textbf{Output:} {$\fPO$ and $\PROPone$ integral allocation $\ialloc = (\ialloc_1, \ialloc_2, \ldots, \ialloc_n)$.}
\begin{itemize}
    \item Let $\dalloc$ be a fractional allocation that maximizes the normalized $p$-mean of utilities, and let $(\dalloc,\prices, \budgets)$ be the corresponding market equilibrium guaranteed to exist by~\Cref{theorem:divisible-goods}.
    \item Let $\ialloc$ be the allocation obtained by rounding the market equilibrium $(\dalloc,\prices, \budgets)$ using the algorithm of~\Cref{theorem:rounding}.
    \item \textbf{return } $\ialloc = \left( \ialloc_1, \ialloc_2, \ldots, \ialloc_n \right)$
\end{itemize}
\caption{$\fPO$ and $\PROPone$ for indivisible goods} \label{algo:goods}
\end{algorithm}

\begin{proof}{Proof of~\Cref{theorem:rounding-goods}}
Let $\dalloc$ be a fractional allocation that maximizes the normalized $p$-mean of utilities, and let $(\dalloc,\prices, \budgets)$ be the corresponding market equilibrium guaranteed to exist by~\Cref{theorem:divisible-goods}. Also, let $\ialloc$ and $\budgets'$ be the integral allocation and budgets guaranteed in~\Cref{theorem:rounding}. First, since $(\ialloc, \prices, \budgets')$ is a market equilibrium, $\ialloc$ is an $\fPO$ allocation (\Cref{theorem:first-welfare}). Towards showing that $\ialloc$ is $\PROPone$, consider the following two exhaustive cases:

    \paragraph{Case 1.} Agents $i$ for which $b'_i < b_i$. We can lower bound the utility of such agents as,
    \begin{align*}
        v_i(\ialloc_i) & = b'_i \cdot \MBB_i \tag{$(\ialloc, \prices,\budgets')$ is a market equilibrium}\\
        & \geq (b_i - p_j) \cdot \MBB_i \tag{for some $j \in \MBB_i \setminus \ialloc_i$,  \Cref{theorem:rounding} property $2$}\\
        & = v_i(\dalloc_i) - v_{ij} \tag{$v_i(\dalloc_i) = b_i \cdot \MBB_i$, and $v_{ij} = p_j \cdot \MBB_i$ since $j \in \MBB_i$}\\
        & \geq \frac{v_i(\items)}{n} - v_{ij}. \tag{$\dalloc$ is $\PROP$}
    \end{align*}
    \paragraph{Case 2.} Agents $i$ for which $b'_i \geq b_i$. For this case, consider the following,
    \begin{align*}
        v_i(\ialloc_i) & = b'_i \cdot \MBB_i \tag{$(\ialloc, \prices, \budgets')$ is a market equilibrium}\\
        & \geq b_i \cdot \MBB_i\\
        & = v_i(\dalloc_i) \geq \frac{v_i(\items)}{n}. \tag{$\dalloc$ is $\PROP$}
    \end{align*}
    Hence, $\ialloc$ is $\PROPone$ in both the cases.  
\end{proof}

\subsection{Indivisible Goods: Optimizing $p$-means}\label{sec:indivisible-goods}

In this section, we prove that directly maximizing the normalized $p$-mean of utilities for indivisible goods gives $\EFone$ and $\PO$ allocations, for the case of $n=2$ agents, for all $p \leq 0$. This is not the case, however, for $n \geq 3$ agents, except, of course, for the case of Nash welfare (i.e., $p = 0$). We use the following technical lemma, which proves that if the minimum and the product of a pair of numbers is higher than another pair, then for any $p\leq 0$, its $p$-mean will also be higher.


\begin{lemma} \label{lemma:indiv_goods_algebra}
    For positive real numbers $a$, $b$, $\alpha$, and $\beta$, if $\min\{a, b\} \leq \min\{\alpha, \beta\}$ and $ab < \alpha\beta$, then $\left(\frac{a^p + b^p}{2}\right)^{1/p} < \left(\frac{\alpha^p + \beta^p}{2}\right)^{1/p}$ for any $p \leq 0$.
\end{lemma}
\begin{proof}   
    The lemma trivially holds for $p = 0$ because $\lim_{p \rightarrow 0} \left(\frac{a^p + b^p}{2}\right)^{1/p} = ab$. Henceforth, we assume that $p < 0$. Without loss of generality, let $a \leq b$ and $\alpha \leq \beta$. Thus, $\min\{a, b\} \leq \min\{\alpha, \beta\}$ implies that $a \leq \alpha$. 
    
    Define $x, y, z \geq 0$ such that
    $\alpha^p = a^p - x$, $b^p = a^p - y$, and
    $\beta^p = \alpha^p - z$, which further gives us $\beta^p = a^p - x - z$. Now, using the fact that $ab < \alpha \beta$ we will derive an inequality relating $x, y$, and $z$. Towards this, note that $ab < \alpha \beta$ implies that $a^pb^p > \alpha^p \beta^p$. Simplifying, we get,

    \begin{align*}
        a^p b^p & > \alpha^p \beta^p \\
        \implies a^p (a^p - y)& > (a^p-x)(a^p - x - z)\\
        \implies  - y a^p & > -a^p(x+z) -x a^p \\
        \implies  - y & > - 2x - z. \tag{$a^p > 0$}
    \end{align*}
    Adding $2 a^p$ on both sides of the above inequality gives us, $a^p + (a^p - y) > (a^p - x - z) + (a^p - x)$, i.e., $a^p + b^p > \beta^p + \alpha^p$. Since $p<0$, this gives us our desired inequality $(a^p + b^p)^{1/p} < (\alpha^p + \beta^p)^{1/p}$.  
\end{proof}

\begin{theorem}\label{theorem:indiv-goods-EF1}
For $n=2$ agents with additive valuations over indivisible goods, and all $p \leq 0$, every allocation that maximizes the normalized $p$-mean of utilities is $\EFone$ and $\PO$.
\end{theorem}
\begin{proof}   
    Pareto optimality holds trivially. 
    Let $\ialloc = (\ialloc_1, \ialloc_2)$ be an allocation that is not $\EFone$. We will show that, for all $p \leq 0$, there always exists another allocation $\ialloc'$ that has a strictly larger normalized $p$-mean of utilities than $\ialloc$; the theorem follows.

    Since $\ialloc$ is not $\EFone$, one of the agents must have normalized utility less than $1/2$. Without loss of generality, suppose that $\tilde{v}_1(\ialloc_1) < \frac{1}{2}$. If it is also the case that $\tilde{v}_2(\ialloc_2) \leq \frac{1}{2}$ then the allocation $\ialloc'$ that is constructed by simply swapping the bundles of the agents, i.e., $\ialloc' = (\ialloc_2, \ialloc_1)$, has a higher normalized $p$-mean of utilities than $\ialloc$: the normalized utility of agent $1$ strictly increases, $\tilde{v}_1(\ialloc'_1) =  \tilde{v}_1(\ialloc_2) = 1 - \tilde{v}_1(\ialloc_1) > \frac{1}{2} > \tilde{v}_1(\ialloc_1)$, and the normalized utility of agent $2$ doesn't decrease $\tilde{v}_2(\ialloc'_1) = \tilde{v}_2(\ialloc_1) = 1 - \tilde{v}_2(\ialloc_2) \geq \frac{1}{2} \geq \tilde{v}_2(\ialloc_2)$. We can therefore assume that $\tilde{v}_2(\ialloc_2) > \frac{1}{2}$. Additionally, since $\ialloc$ is not EF1, we have that for all items $g \in \ialloc_2$, $v_1(\ialloc_1) < v_1(\ialloc_2) - v_1(g)$, which implies that for all items $g \in \ialloc_2$, $\tilde{v}_1(\ialloc_1) < \tilde{v}_1(\ialloc_2) - \tilde{v}_1(g)$. 
    
    We will first consider the case when $\tilde{v}_1(\ialloc_1) = 0$; note that, here the normalized $p$-mean of utilities is also zero. Define $V = \{j \in \ialloc_2 \ | \ \tilde{v}_1(j) > 0\}$ to be the set of items in $\ialloc_2$ that are positively valued by agent $1$. Since $\ialloc$ is not $\EFone$, agent $1$ must value at least $2$ items in $\ialloc_2$, i.e., $|V| \geq 2$. Transferring an item $g \in \arg \min\limits_{j \in V} \{\tilde{v}_2(j)\}$ which is valued least by agent $2$ from $\ialloc_2$ to $\ialloc_1$ will $(1)$ make agent $1$'s normalized utility positive, and $(2)$ will ensure that agent $2$'s normalized utility stays positive, since $g$ is her least-valued item in $V$ and $|V| \geq 2$. Thus, the transfer strictly increases the normalized $p$-mean of utilities: both agents' normalized utilities are positive after this transfer, resulting in a positive normalized $p$-mean of utilities, whereas, before the transfer, the $p$-mean was zero. In the subsequent proof, we will consider the remaining case where $\tilde{v}_1(\ialloc_1) > 0$.

    Define $S \subseteq \ialloc_2$ to be the following set: 
    $S \coloneqq \{g \in \ialloc_2 \: | \: \tilde{v}_1(\ialloc_1)\cdot \tilde{v}_2(\ialloc_2) < \tilde{v}_1(\ialloc_1 \cup \{ g \} )\cdot\tilde{v}_2(\ialloc_2 \setminus \{ g \})\}$.
    Intuitively, $S$ is the set of all goods $g$ in $\ialloc_2$ such that transferring $g$ from agent $2$ to agent $1$ strictly increases the (normalized) Nash Welfare, i.e. the product of agents' normalized utilities. 
    
    Since $\ialloc$ is not EF1, we know that transferring an item from $\ialloc_2$ to $\ialloc_1$ will increase the Nash welfare~\cite{caragiannis2019unreasonable}, implying that $S$ is nonempty. In particular, every good $g \in \arg \max_{j \in \ialloc_2} \frac{\tilde{v}_1(j)}{\tilde{v}_2(j)}$ must be in $S$. To see this, notice that for such a good $g$, by definition we have $\frac{\tilde{v}_1(g)}{\tilde{v}_2(g)} \geq \frac{\tilde{v}_1(\ialloc_2)}{\tilde{v}_2(\ialloc_2)}$, which along with the fact that $\tilde{v}_1(\ialloc_1) < \tilde{v}_1(\ialloc_2) - \tilde{v}_1(g)$ implies that $\tilde{v}_1(g)\tilde{v}_2(\ialloc_2) > (\tilde{v}_1(\ialloc_1) + \tilde{v}_1(g))\tilde{v}_2(g)$. Adding $\tilde{v}_1(\ialloc_1)\tilde{v}_2(\ialloc_2)$ to both sides of this inequality and re-arranging gives us $(\tilde{v}_1(\ialloc_1) + \tilde{v}_1(g))(\tilde{v}_2(\ialloc_2) - \tilde{v}_2(g)) > \tilde{v}_1(\ialloc_1)\tilde{v}_2(\ialloc_2)$, showing that $g \in S$, i.e, $S$ is non-empty.

    Consider an arbitrary good $g \in S$. We consider two exhaustive cases based on whether $\tilde{v}_1(\ialloc_1) > \tilde{v}_2(\ialloc_2) - \tilde{v}_2(g)$ is true or not. In both cases, we construct an allocation $\ialloc'$ having a higher normalized $p$-mean for every $p \leq 0$.

    \paragraph{Case 1. $\tilde{v}_1(\ialloc_1) > \tilde{v}_2(\ialloc_2) - \tilde{v}_2(g)$.} In this case, we will show that the normalized $p$-mean of $\ialloc' = (\ialloc_2 \setminus \{g\}, \ialloc_1 \cup \{g\})$, obtained by transferring $g$ from $\ialloc_2$ to $\ialloc_1$ and then swapping the bundles, is higher than the normalized $p$-mean of $\ialloc$. Towards this, we will show that $(i)$ the normalized Egalitarian welfare (i.e., the minimum utility) of agents in $\ialloc'$ is larger than $\ialloc$, and $(ii)$ the normalized Nash welfare of $\ialloc'$ is strictly larger than $\ialloc$. Note that $(i)$ and the fact that $\tilde{v}_2(\ialloc_2) > 1/2 > \tilde{v}_1(\ialloc_1) > 0$ implies that $\tilde{v}_1(\ialloc_1),\tilde{v}_2(\ialloc_2),\tilde{v}_1(\ialloc_2\setminus \{g\})$, $\tilde{v}_2(\ialloc_1\cup \{g\})$ are all strictly positive. Thus, given $(i)$ and $(ii)$, we can use~\Cref{lemma:indiv_goods_algebra} to conclude that, for all $p \leq 0$, the normalized $p$-mean of utilities in $\ialloc'$ is strictly larger than in $\ialloc$.

    To prove $(i)$, first note that $\tilde{v}_2(\ialloc_2) > 1/2 > \tilde{v}_1(\ialloc_1)$, i.e., $\min\{\tilde{v}_1(\ialloc_1), \tilde{v}_2(\ialloc_2)\} = \tilde{v}_1(\ialloc_1)$. Second, we have $\tilde{v}_1(\ialloc_1) < \tilde{v}_1(\ialloc_2 \setminus \{g\})$, since $\ialloc$ is not $\EFone$. Additionally, we can show that $\tilde{v}_1(\ialloc_1) < \tilde{v}_2(\ialloc_1 \cup \{g\})$ as follows:
    \begin{align*}
        \tilde{v}_2(\ialloc_1 \cup \{g\}) & = 1 - (\tilde{v}_2(\ialloc_2) - \tilde{v}_2(g))\tag{$\tilde{v}_1(\ialloc_1) + \tilde{v}_1(\ialloc_2) = 1$}\\
        & > 1 - \tilde{v}_1(\ialloc_1) \tag{$\tilde{v}_1(\ialloc_1) > \tilde{v}_2(\ialloc_2) - \tilde{v}_2(g)$}\\
        & > \tilde{v}_1(\ialloc_1). \tag{$1/2 > \tilde{v}_1(\ialloc_1)$}
    \end{align*}
    The inequalities proved above imply $(i)$ since $\min\{\tilde{v}_1(\ialloc_1), \tilde{v}_2(\ialloc_2)\} = \tilde{v}_1(\ialloc_1) < \min\{\tilde{v}_1(\ialloc_2\setminus \{g\}), \tilde{v}_2(\ialloc_1 \cup \{g\})\}$. For proving $(ii)$, we consider the following sequence of inequalities,

    \begin{align*}
        &\tilde{v}_1(\ialloc_2 \setminus \{g\}) \cdot \tilde{v}_2(\ialloc_1 \cup \{g\}) - \tilde{v}_1(\ialloc_1)\cdot \tilde{v}_2(\ialloc_2)\\
        &= (1 - \tilde{v}_1(\ialloc_1 \cup \{g\})(1 - \tilde{v}_2(\ialloc_2 \setminus \{g\})) - \tilde{v}_1(\ialloc_1)\cdot \tilde{v}_2(\ialloc_2) \tag{$\tilde{v}_i(\ialloc_1) + \tilde{v}_i(\ialloc_2) = 1$}\\
        &= \Big( \tilde{v}_1(\ialloc_1 \cup \{g\})\cdot \tilde{v}_2(\ialloc_2 \setminus \{g\}) - \tilde{v}_1(\ialloc_1)\cdot \tilde{v}_2(\ialloc_2) \Big) + 1 - \tilde{v}_1(\ialloc_1 \cup \{g\}) - \tilde{v}_2(\ialloc_2 \setminus \{g\})\\
        &> 0 + 1 - \tilde{v}_1(\ialloc_1 \cup \{g\}) - \tilde{v}_2(\ialloc_2 \setminus \{g\}) \tag{$g \in S$}\\
        & = (\tilde{v}_1(\ialloc_2) - \tilde{v}_1(g)) - (\tilde{v}_2(\ialloc_2) - \tilde{v}_2(g)) \tag{$\tilde{v}_1(\ialloc_1) + \tilde{v}_1(\ialloc_2) = 1$}\\
        & > \tilde{v}_1(\ialloc_1) - (\tilde{v}_2(\ialloc_2) - \tilde{v}_2(g)) \tag{$\tilde{v}_1(\ialloc_1) < \tilde{v}_1(\ialloc_2) - \tilde{v}_1(g)$}\\
        & > 0. \tag{$\tilde{v}_1(\ialloc_1) > \tilde{v}_2(\ialloc_2) - \tilde{v}_2(g)$}
    \end{align*}


    \paragraph{Case 2. $\tilde{v}_1(\ialloc_1) \leq \tilde{v}_2(\ialloc_2) - \tilde{v}_2(g)$.} We show that the allocation $\ialloc' = (\ialloc_1 \cup \{g\}, \ialloc_2 \setminus \{g\})$ has a higher Nash and normalized Egalitarian welfare than $\ialloc$; applying~\Cref{lemma:indiv_goods_algebra} completes the argument. 

    Since $g \in S$, by definition, transferring $g$ from $\ialloc_2$ to $\ialloc_1$ strictly increases the (normalized) Nash welfare. Considering the normalized Egalitarian welfare, we again have $0 < \tilde{v}_1(\ialloc_1)  = \min\{\tilde{v}_1(\ialloc_1), \tilde{v}_2(\ialloc_2)\}$. Additionally, since $\tilde{v}_1(g) > 0$, we have $\tilde{v}_1(\ialloc_1) < \tilde{v}_1(\ialloc_1 \cup \{g\})$ and $\tilde{v}_1(\ialloc_1) \leq \tilde{v}_2(\ialloc_2 \setminus \{g\})$. Chaining these inequalities gives us $0 < \min\{\tilde{v}_1(\ialloc_1), \tilde{v}_2(\ialloc_2)\} = \tilde{v}_1(\ialloc_1) \leq \min\{\tilde{v}_1(\ialloc_1 \cup \{g\}), \tilde{v}_2(\ialloc_2 \setminus \{g\})\}$. Invoking~\Cref{lemma:indiv_goods_algebra} concludes the proof.  
\end{proof}

Next, we prove that, even for $n=2$ agents,~\Cref{theorem:indiv-goods-EF1} does not hold for $p > 0$, and that for $n > 2$ agents, maximizing the normalized $p$-mean of utilities is not even $\PROPone$ (hence not $\EFone$) for any $p \neq 0$.  

\begin{theorem}\label{theorem:indiv-goods-tight-wrt-n}
For $n > 2$ agents with additive valuations over indivisible goods, and all  $p < 0$, there exist instances such that every allocation that maximizes the normalized $p$-mean of utilities is not $\PROPone$. Furthermore, for all $p > 0$ and $n \geq 2$ agents, there exist instances such that every allocation that maximizes the normalized $p$-mean of utilities is not $\PROPone$.
\end{theorem}
\begin{proof}   
The second part of the theorem, for $p > 0$, follows from~\Cref{lemma:div-to-indiv} and~\Cref{theorem:div-goods-negative}, since if all allocations that maximize the normalized $p$-mean of utilities are $\PROPone$, then~\Cref{lemma:div-to-indiv} would imply that allocations of divisible goods maximizing the $p$-mean for $p>0$ would be $\PROP$ as well, violating~\Cref{theorem:div-goods-negative}.

For the first part of the theorem, consider the following instance with $n=3$ agents and $m=7$ indivisible goods, where $\varepsilon > 0$ is a small number:
    $$\begin{pmatrix}
        1/7 ~~& 1/7 ~~& 1/7 ~~& 1/7 ~~& 1/7 ~~& 1/7 ~~& 1/7 \\
        \varepsilon/6 & \varepsilon/6 & \varepsilon/6 & \varepsilon/6 & \varepsilon/6 & \varepsilon/6 & 1-\varepsilon\\
        0 & 0 & 0 & 0 & 0 & 0 & 1
    \end{pmatrix}.$$
    For every $p < 0$ and any allocation that maximizes the normalized $p$-mean of utilities, agent 3 must get the last item, since she will otherwise have zero utility, and thus the $p$-mean will be equal to zero. Additionally, it is easy to verify that for every $p < 0$, one can pick $\varepsilon > 0$ sufficiently small so that agent $1$ will only get a single good in any $p$-mean maximizing allocation, resulting in a utility of $1/7$ for agent $1$. Such an allocation is not $\PROPone$, since agent $1$'s utility after receiving one more good is $2/7$, which is strictly less than her proportional share, $1/3$. 
    One can extend the above construction to any $n > 3$ by adding a new item $i^*$ and a corresponding agent $i$ such that $i^*$ is the only item $i$ likes, and $i^*$ is liked only by agent $i$.  
\end{proof}

\section{Allocating Chores}\label{SEC:CHORES}

In this section, we present our results for the case of chores. For ease of notation, for the remainder of this section, we use $\tilde{c}_i$ to refer to the \emph{normalized} cost of agent $i$, i.e. $\tilde{c}_{ij} = c_{ij}/c_i(\items)$ and $\tilde{c}_{i}(\dalloc_i) = c_{i}(\dalloc_i)/c_i(\items)$.
For ease of exposition, we say that we minimize the \emph{normalized $p$-mean} of disutilities to mean that we minimize the $p$-mean of normalized disutilities, i.e., find the allocation $\dalloc$ with minimum $\mean_p( \tilde{c}_{i}(\dalloc_i), \ldots,  \tilde{c}_{n}(\dalloc_n) )$.

In~\Cref{sec:chores negative result}, we begin by proving impossibility results for welfarist rules. For divisible chores, we show that welfarist rules (see~\Cref{sec:chores negative result} for the definition) cannot result in $\bEF$ or $\bPROP$ allocations; this sharply contrasts with the case of goods, where maximizing Nash welfare produces $\EF$ allocations. Similar impossibility is established for indivisible chores: optimizing a welfarist rule cannot result in a $\bEFk$ or $\bPROPk$ allocation; this again contrasts with the case of goods, where optimizing the Nash social welfare is $\EFone$ (and, in fact, this is the unique welfarist rule with this property~\cite{yuen2023extending}). 

In~\Cref{sec: divisible chores} we prove that, for $p \geq 1$, minimizing the normalized $p$-mean of disutilities for divisible chores gives $n^{1/p}$-$\PROP$ and $\fPO$ allocations, and show that, similar to the case of goods, the KKT conditions of the corresponding convex program allow for a market interpretation of the optimal solution; we prove that our approximation on proportionality is tight.  Building upon this, in~\Cref{sec:chores rounding} we show that it is possible to round the fractional allocation that minimizes the normalized $p$-mean, for $p \geq 1$, to get integral allocations that are $\fPO$ and approximately $\PROPone$. Finally, in~\Cref{sec:indivisible chores} we prove that, for the case of two agents and $p \geq 2$, integral allocations that minimize the normalized $p$-mean of disutilities are $\EFone$ and $\PO$ (and this is not true for more than two agents or $p < 2$).

\subsection{Chore Division: Impossibilities for Welfarist Rules}\label{sec:chores negative result}

In this section, we prove that in sharp contrast with the case of goods, welfarist rules for divisible chores cannot achieve $\bEF$ for any $\beta \geq 1$ or $\bPROP$ for any $\beta \in [1,n)$; indeed, given $n$ agents, $n$-$\PROP$ is trivially achieved by every divisible allocation. Along with~\Cref{lemma:div-to-indiv}, this impossibility result directly implies that for indivisible chores, welfarist rules cannot result in $\bEFk$ allocations for any $\beta \geq 1$ and $k \geq 1$ or $\bPROPk$ allocations for any $\beta \in [1, n)$ and $k \geq 1$.


 \begin{definition}[Welfarist rule]
 Let $\calC:\mathbb{R}^n_{\geq 0} \mapsto \mathbb{R}$ be a weakly-increasing cost function.\footnote{A function $g$ is weakly-increasing if $g(x) > g(y)$ when $x_i > y_i$ for all $i \in [n]$.} Given an instance, a welfarist rule with cost function $\calC$ chooses an allocation $\ialloc$ that minimizes $\calC( c_1(\ialloc_1), \allowbreak \dots, c_n(\ialloc_n) )$. If there are multiple such allocations, the rule may choose an arbitrary one.
\end{definition}

A welfarist rule is $\bEF$ ($\bPROP$) if it always outputs a $\bEF$ ($\bPROP$) allocation. We now state and prove our impossibility result for $\bEF$; the result for $\bPROP$ follows as a corollary.

\begin{theorem}\label{theorem:non-norm-neg-res}
Even for $n=2$ agents with additive costs, there does not exist a weakly-increasing cost function $\calC:\mathbb{R}^n_{\geq 0} \mapsto \mathbb{R}$ such that, for all instances with divisible chores, the welfarist rule with cost function $\calC$ is $\bEF$, for any $\beta \geq 1$. 
\end{theorem}

\begin{proof}   
Towards a contradiction, assume that there exists a weakly-increasing cost function $\calC$ such that for any indivisible chore division instance and agents with additive costs, one of the allocations that minimizes $\calC$ is $\bEF$, for a given $\beta \geq 1$. Since $\bEF$ implies $\beta'\text{-}\mathrm{EF}$ for $\beta' > \beta$, without loss of generality we can assume that $\beta$ is an integer. In the subsequent proof we analyze two divisible chore division instances, $\calI_1$ and $\calI_2$, and show that assuming that the aforementioned cost function $\calC$ exists leads to a contradiction. 

\paragraph{Instance $\calI_1$ and its analysis.}
Instance $\calI_1$ consists of $n=2$ agents and $m$ divisible chores where $m = 2\beta+2$. Agent $1$ has a cost function $c_1$ such that $c_{1j} = 1$ for all chores $j\in \items$ and agent $2$ has cost function $c_2$ satisfying $c_{21} = 1$ and $c_{2j} = m$ for all chores $j \in \items \setminus \{1\}$. We first prove the following claim about instance $\calI_1$.

\begin{claim}\label{claim:non-norm-EF1}
The disutilities of the agents in any $\bEF$ allocation $\dalloc = (\dalloc_1, \dalloc_2)$ of the instance $\calI_1$ satisfy $c_1(\dalloc_1) \geq 1.5$ and $c_2(\dalloc_2) \geq m + 1$.
\end{claim}
\begin{proof}   
Let $\dalloc = (\dalloc_1, \dalloc_2)$ be a $\bEF$ allocation of instance $\calI_1$. The $\bEF$ constraint for agent $1$ implies that $c_1(\dalloc_1) \leq \beta c_1(\dalloc_2)$. Since $c_1(\dalloc_1) + c_1(\dalloc_2) = c_1(\items) = m$, we have $m-c_1(\dalloc_2) \leq \beta c_1(\dalloc_2)$. Therefore, $c_1(\dalloc_2) \geq \frac{m}{\beta + 1} = \frac{2\beta + 2}{\beta + 1} = 2$. Since agent $1$ values all chores at $1$, $c_1(\dalloc_2) \geq 2$ implies that $\dalloc_2$ satisfies $\sum_{j \in \items} x_{2,j} \geq 2$. Since agent $2$ values one chore at $1$ and the rest at $m$, we have $c_2(\dalloc_2) \geq m + 1$. 

It remains to show that $c_1(\dalloc_1) \geq 1.5$. To this end, we consider agent $2$'s $\bEF$ guarantee, $c_2(\dalloc_2) \leq \beta c_2(\dalloc_1)$. Along with the fact that $c_2(\dalloc_1) + c_2(\dalloc_2) = m(m-1)+1$, this implies that $c_2(\dalloc_1) \geq \frac{m(m-1)+1}{\beta + 1} = \frac{m(m-1)+1}{m/2} = 2(m-1) + \frac{2}{m}$. Since the cost of any chore for agent $2$ is at most $m$, the previous inequality implies that agent $1$ must be getting at least $\frac{c_2(\dalloc_1)}{m} \geq 2 - \left( \frac{2}{m} - \frac{2}{m^2} \right)$ unit of chores. The maximum value of $2 - \left( \frac{2}{m} - \frac{2}{m^2} \right)$ is reached when its derivative is zero, i.e., when $\frac{2}{m^2} - \frac{4}{m^3} = 0$ or at $m=2$. Thus, agent $1$ gets at least $2-\left( \frac{2}{2} - \frac{2}{2^2} \right) = 1.5$ units of chores, which in turn implies that, $c_1(\dalloc_1) \geq 1.5$.  
\end{proof}

Consider the allocation $\dalloc = (\dalloc_1, \dalloc_2)$ where $\dalloc_1 = \items \setminus \{1\}$ and $\dalloc_2 = \{1\}$, for which $c_1(\dalloc_1) = m-1$ and $c_2(\dalloc_2) = 1$. Let $\mathcal{F} = \{ \dalloc = (\dalloc_1, \dalloc_2)  : \ \dalloc \text{ is } \bEF \}$ denote the set of all $\bEF$ allocations of instance $\calI_1$. By definition, the global minimum of $\calC$ must occur at some $\bEF$ allocation. Thus, we have
\begin{align*}
    \calC(m-1, 1) &\geq \min_{(\dalloc_1, \dalloc_2) \in \mathcal{F}}  \calC(c_1(\ialloc_1), c_2(\ialloc_2)) \\
    &\geq \calC(1.5, m + 1) \tag{\Cref{claim:non-norm-EF1}} \\
    &> C(1,m-1). \tag{$\calC$ is a weakly-increasing function}
\end{align*}

\paragraph{Instance $\calI_2$ and its analysis.}
Consider an instance $\calI_2$ with $n=2$ agents and $m=2\beta + 2$ chores such that agent $1$ has the cost function $c_2$ and agent $2$ has the cost function $c_1$ (i.e., $\calI_2$ is the instance obtained by swapping the names of agents $1$ and $2$ in $\calI_1$). Then, by using exactly the same arguments as in $\calI_1$, we can derive the inequality $\calC(1,m-1) > \calC(m-1,1)$; this contradicts the previously obtained inequality.  
\end{proof}

\begin{corollary}\label{corollary:neg-result-bPROP}
Even for $n=2$ agents with additive costs, there does not exist a weakly-increasing cost function $\calC:\mathbb{R}^n_{\geq 0} \mapsto \mathbb{R}$ such that, for all instances with divisible chores, the welfarist rule with cost function $\calC$ is $\bPROP$, for any $\beta \in [1,2)$. 
\end{corollary}

\begin{proof}   
    We will show that for any $2$ agent instance, if an allocation $\dalloc$ is $\bPROP$ for $\beta \in [1,2)$, then is $\frac{\beta}{2-\beta}$-$\EF$ as well. Note that, for any fixed $\beta \in [1,2)$, the ratio $\frac{\beta}{2-\beta}$ lies is in the range $[1, \infty)$. The corollary directly follows from this implication, since if $\calC$ achieves $\bPROP$, then it would also achieve $\frac{\beta}{2-\beta}$-$\EF$, contradicting~\Cref{theorem:non-norm-neg-res}.
    Towards proving the implication, consider a $\bPROP$ allocation $\dalloc$. Since $\dalloc$ is $\bPROP$, for agent $i$ we have, $c_i(\dalloc_i) \leq \beta \cdot \frac{c_i(\items)}{2} = \frac{\beta}{2} \cdot (c_i(\dalloc_i) + c_i(\dalloc_{-i}))$, where $-i$ is the remaining agent. On rearranging, we get $c_i(\dalloc_i) \leq \frac{\beta}{2-\beta} \cdot c_i(\dalloc_{-i})$, implying that $\dalloc$ is $\frac{\beta}{2-\beta}$-$\EF$. This completes the proof.  
\end{proof}

The following corollary directly follows from~\Cref{lemma:div-to-indiv},~\Cref{theorem:non-norm-neg-res}, and~\Cref{corollary:neg-result-bPROP}.

\begin{corollary}\label{corollary:neg-results-indivisible-chores}
Even for $n=2$ agents with additive costs, there does not exist a weakly-increasing cost function $\calC:\mathbb{R}^n_{\geq 0} \mapsto \mathbb{R}$ such that, for all instances with indivisible chores, the welfarist rule with cost function $\calC$ is $\bEFk$, for any $\beta \geq 1$ and $k \geq 1$ (resp. $\bPROPk$ for any $\beta \in [1, 2)$ and $k \geq 1$).
\end{corollary}

\subsection{Divisible Chores}\label{sec: divisible chores}

We first prove our positive result for divisible chores.

\begin{theorem} \label{theorem:div-chore-prop}
For any number of agents with additive costs over divisible chores, and all $p \geq 1$, every allocation $x^*$ that minimizes the normalized $p$-mean of disutilities is $n^{1/p}$-$\PROP$ and $\fPO$. Furthermore, there exist rewards $r^*_j \geq 0$ for each item $j \in \items$ and earnings $e^*_i$ for each agent $i \in \agents$, such that $(\dalloc^*, \rewards^*, \earnings^*)$ is a market equilibrium.
\end{theorem}
\begin{proof}   
    Fix an arbitrary $p \geq 1$, and let $\dalloc^* = (\dalloc^*_1, \ldots, \dalloc^*_n)$ be any allocation that minimizes the $p$-mean of agents' normalized disutilities. The fact that $\dalloc^*$ is $\fPO$ is clear; next, we prove that it is $n^{1/p}$-$\PROP$. Since $\dalloc^*$ minimizes the $p$-mean of normalized disutilities, we have that $\sum_{i=1}^n \tilde{c}_i(\dalloc^*_i)^p$ must be at most $\sum_{i=1}^n \tilde{c}_i(\dallocy_i)^p$, where $y_{ij} = \frac{1}{n}$ for all agent $i \in \agents$ and items $j \in \items$. Therefore,
    \begin{align*}
         \frac{n}{n^p} = \sum_{i=1}^n \tilde{c}_i(\dallocy_i)^p \geq \sum_{i=1}^n \tilde{c}_i(\dalloc^*_i)^p \geq \tilde{c}_i(\dalloc^*_i)^p.
    \end{align*}
    Taking the $p^{\text{th}}$ root on both sides gives us the desired inequality: for every agent $i \in \agents$, we get $\tilde{c}_i(\dalloc^*_i) \leq \frac{n^{1/p}}{n}$, and thus $c_i(\dalloc^*_i) \leq n^{1/p} \, \frac{\sum_{j \in \items} c_{ij}}{n} $, i.e., $\dalloc$ is $n^{1/p}$-$\PROP$.

    The proof of the second part of the theorem, i.e. the fact that one can find rewards and earnings to get a market equilibrium, is similar to the arguments in~\Cref{theorem:divisible-goods}. The proof is deferred to Appendix~\ref{app: missing from chores}.  
\end{proof}

The following corollary simply observes that by picking $p \in \Omega(\log{n})$,~\Cref{theorem:div-chore-prop} implies that minimizing the normalized $p$-mean of disutilities gives a constant approximation to proportionality; it follows from the fact that $n^{1/\log{n}} = \Theta(1)$.

\begin{corollary}
For any number of agents with additive costs over divisible chores, and $p \in \Omega(\log{n})$, every allocation that minimizes the normalized $p$-mean of disutilities is $\fPO$ and guarantees a constant factor approximation of $\PROP$.
\end{corollary}

The proportionality guarantee in~\Cref{theorem:div-chore-prop} was established by a rather simple argument. Therefore, one would perhaps expect that it is possible to improve our analysis and get a better approximation to proportionality. Surprisingly, we prove that our guarantee is essentially tight, i.e., there are instances where all the allocations that minimize the normalized $p$-mean of disutilities cannot be $\beta$-$\PROP$ for $\beta = n^{1/p}(1-\omega(\frac{\log{p}}{p}))$.

\begin{theorem} \label{theorem:chores-div-neg-many-agents}
For all $p \in \mathbb{R}$, there exists $n \geq 2$ and an instance with $n$ agents with additive costs over divisible chores, such that every allocation that minimizes the normalized $p$-mean of disutilities is not $\PROP$. 
Additionally, for $p>1$, there exist instances with $n>p$ agents where every $p$-mean minimizing allocation is not $\beta$-$\PROP$, for $\beta <  (n^{1/p}\cdot f(n,p))$, where the function $f(n,p)$ satisfies $\lim\limits_{n \rightarrow \infty} f(n,p) = 1 - \Theta(\frac{\log{p}}{p})$.
\end{theorem}
\begin{proof}   
    We consider three cases, based on the value of $p$.

    \paragraph{Case 1. $p < 1$.} Consider a chore division instance with two agents and $m$ chores such that $c_{ij} = 1/m$ for all $i,j$. Note that for all $p<1$, the normalized $p$-mean of disutilities is uniquely minimized when all the chores are assigned to one of the agents, i.e., either $c_1(\dalloc_1) = 1$ and $c_2(\dalloc_2) = 0$ or $c_1(\dalloc_1) = 0$ and $c_2(\dalloc_2) = 1$. Clearly, in both these allocations, $\dalloc$ does not satisfy $\PROP$.
    
    
    \paragraph{Case 2. $p = 1$.} Consider a chore division instance with two agents and two chores such that $c_{11} < c_{21} < 1/2$, $c_{12} > c_{22} > 1/2$, and $ \sum_{j=1}^2 c_{ij} = 1$ for every agent $i$. Indeed, any fractional allocation minimizing the normalized $1$-mean of disutilities must allocate chore $1$ entirely to agent $1$ and chore $2$ to agent $2$. This allocation is not proportional for agent $2$ since $c_{22} > v(\items)/2 = 1/2$.

    \paragraph{Case 3. $p > 1$.} We now show the second part of the theorem, i.e. that if $p > 1$, there exist instances with $n>p$ agents where every $p$-mean minimizing allocation is not $\beta$-$\PROP$ for $\beta < (n^{1/p}\cdot f(n,p))$.
    Consider the following instance with $m=2$ chores: the costs of agent $1$ are $c_{11} = 1 - \left( \frac{p-1}{n-1} \right)^{\frac{p-1}{p}}$ and $c_{12} = \left( \frac{p-1}{n-1} \right)^{\frac{p-1}{p}}$. The remaining $n-1$ agents are identical; agent $j \in \agents \setminus \{1\}$ has costs $c_{j1} = 0$ and $c_{j2} = 1$. Note that, since $n > p$, $c_{11} = 1 - \left( \frac{p-1}{n-1} \right)^{\frac{p-1}{p}} \geq 0$ (i.e. the costs are non-negative and the instance is valid).

    Let $\dalloc = (\dalloc_1, \ldots, \dalloc_n)$ be an allocation that minimizes the $p$-mean of normalized disutilities. First, note that chore $1$ must be fully allocated among agents in $\agents \setminus \{1\}$, since they all have $0$ cost for it. Additionally, in $\dalloc$ all agents $j \in \agents \setminus \{1\}$ will receive an equal amount of chore $2$, i.e., $x_{j2} = x_{k2}$ for all $j,k \in \agents \setminus \{1\}$. Otherwise if $x_{j2} > x_{k2}$, then transferring a small amount of chore $2$ from $j$ to $k$ will decrease the normalized $p$-mean of disutilities since $\frac{d}{d \epsilon} (x_{j2} - \epsilon)^p + (x_{k2} + \epsilon)^p < 0$ if $p>1$.
     Therefore, without loss of generality we can assume $x_{12} = x$ and $x_{j2} = \frac{1-x}{n-1}$ for all $j \in \agents \{1\}$. This implies that $\tilde{c}_1(\dalloc_1) = x \left( \frac{p-1}{n-1} \right)^{\frac{p-1}{p}}$ and $\tilde{c}_j(\dalloc_j) = \frac{1-x}{n-1}$ for all $j \in \agents \setminus \{1\}$. Since $\dalloc$ minimizes the $p$-mean of normalized disutilities, we can find $x$ by computing $\frac{d}{dx} \sum_{i=1}^n \tilde{c}_i(\dalloc_i)^p$ and setting it to zero. Note that,
    
    \begin{align*}
        \frac{d}{dx} \left( \sum_{i=1}^n \tilde{c}_i(\dalloc_i)^p \right) & = px^{p-1} \left( \frac{p-1}{n-1} \right)^{p-1} - (n-1) \cdot \frac{p(1-x)^{p-1}}{(n-1)^{p}}\\
        & = \frac{p}{(n-1)^{p-1}} \left(x^{p-1} (p-1)^{p-1} - (1-x)^{p-1} \right).
    \end{align*}
    Setting this to zero, gives us $x(p-1) - (1-x) = 0$ or $x = 1/p$. Furthermore, the derivative is negative for $x < 1/p$ and positive for $x > 1/p$, and thus $x=1/p$ corresponds to the minimum value of the $p$-mean. Thus, the cost of agent $1$ in allocation $\dalloc$ is $\tilde{c}_1(\dalloc_1) = \frac{1}{p}\left( \frac{p-1}{n-1} \right)^{\frac{p-1}{p}} = \frac{n^{1/p}}{n} \cdot f(n,p)$ where $f(n,p) = \frac{n}{n-1} \cdot \frac{p-1}{p} \cdot \left( \frac{n-1}{n(p-1)} \right)^{1/p}$. 
    First, $\tilde{c}_1(\dalloc_1) > 1/n$, i.e. $\dalloc$ is not $\PROP$: this is equivalent to proving that $\frac{n}{(n-1)^{1-1/p}} > \frac{p}{(p-1)^{1-1/p}}$, which follows from the fact that $(i)$ the function $g(y) = \frac{y}{(y-1)^{1-1/p}}$ is increasing for $y \geq p$ and $(ii)$ $n>p$, and therefore $g(n) > g(p)$ or $\frac{n}{(n-1)^{1-1/p}} > \frac{p}{(p-1)^{1-1/p}}$.
    Second, for any fixed $p > 1$, the limit $\lim\limits_{n \rightarrow \infty} f(n,p) = \frac{(p-1)^{1-1/p}}{p} = 1 - \frac{p-(p-1)^{1-1/p}}{p} = 1 - \Theta(\frac{\log{p}}{p})$, where the value $\frac{(p-1)^{1-1/p}}{p} \geq 0.5$ for all $p>1$, and approaches $1$ as $p$ becomes larger. Thus, every normalized $p$-mean maximizing allocation is exactly $n^{1/p} \left(1 - \Theta(\frac{\log{p}}{p})\right)$-$\PROP$ in this instance.   
\end{proof}

\Cref{theorem:chores-div-neg-many-agents} implies that, for any fixed value of $p > 1$, as $n$ becomes larger, all allocations that minimize the normalized $p$-mean of disutilities cannot be $\beta$-$\PROP$ for $\beta = n^{1/p}(1-\omega(\frac{\log{p}}{p}))$, i.e., the bound in \Cref{theorem:div-chore-prop} is tight. Note that, for the case of $n=2$ agents and $p \geq 2$, allocations that minimize the normalized $p$-mean of disutilities are $\PROP$, as a corollary to~\Cref{theorem:chore-EF1-result} and~\Cref{lemma:div-to-indiv}; therefore, the ``$n > p$'' part of~\Cref{theorem:chores-div-neg-many-agents} also cannot be significantly improved upon either.

\subsection{Indivisible Chores: Rounding Divisible Allocations}\label{sec:chores rounding}

Here, we prove that the market interpretation of $p$-mean optimal solutions for divisible chores (\Cref{theorem:div-chore-prop}) can be leveraged to give new $\fPO$ and $n^{1/p}$-$\PROPone$ algorithms for indivisible chores, for all $p \geq 1$. We use the following result of Br{\^a}nzei and Sandomirskiy~\cite{branzei2023algorithms}.

\begin{theorem}[Theorem 7.2; \cite{branzei2023algorithms}, modification of~\cite{barman2019proximity}]\label{theorem:rounding-for-chores}
Given any market equilibrium for divisible chores $(\dalloc, \rewards, \earnings)$, we can, in polynomial time, compute a market equilibrium $(\ialloc, \rewards, \earnings')$ where $\ialloc$ is an integral allocation such that if $e'_i > e_i$, then there exists a chore $j \in \ialloc_i$ such that $e'_i - r_j \leq e_i$.
\end{theorem}

The proof of the following theorem is similar to the proof of~\Cref{theorem:rounding-goods}.

\begin{theorem}\label{theorem:rounding-chores}
For $n \geq 2$ agents with additive costs over indivisible chores, and all $p \geq 1$,~\Cref{algo:chores} outputs an $\fPO$ and $n^{1/p}$-$\PROPone$ allocation.
\end{theorem}
\begin{proof}   
Let $\dalloc$ be a fractional allocation that minimizes the $p$-mean of normalized disutilities, and let $(\dalloc,\rewards, \earnings)$ be the corresponding market equilibrium guaranteed to exist by~\Cref{theorem:div-chore-prop}. Let $\ialloc$ and $\earnings'$ be the integral allocation and budgets guaranteed in~\Cref{theorem:rounding-for-chores}. First, since $(\ialloc, \rewards, \earnings')$ is a market equilibrium, $\ialloc$ is an $\fPO$ allocation (\Cref{theorem:first-welfare}). Towards showing that $\ialloc$ is $\alpha$-$\PROPone$, consider the following two cases:

    \paragraph{Case 1.} Agents $i \in \agents$ for which $e'_i \leq e_i$. We can upper bound the costs of such agents as,
    \begin{align*}
        c_i(\ialloc_i) & = \frac{e'_i}{\MRC_i} \tag{$(\ialloc, \rewards, \earnings')$ is a market equilibrium}\\
        & \leq \frac{e_i}{\MRC_i} \tag{$e'_i \leq e_i$} \\
        & = c_i(\dalloc_i) \tag{$(\dalloc, \rewards, \earnings)$ is a market equilibrium}\\
        & \leq \frac{n^{1/p} \, c_i(\items)}{n}. \tag{$\dalloc$ is $n^{1/p}$-$\PROP$}
    \end{align*}
    \paragraph{Case 2.} Agents $i \in \agents$ for which $e'_i > e_i$. Here, we will use the fact that there exists a chore $j \in \ialloc_i$ such that $e'_i - r_j \leq e_i$ (\Cref{theorem:rounding-for-chores}). Using this we get,

    \begin{align*}
        c_i(\ialloc_i) & = \frac{e'_i}{\MRC_i}  \tag{$(\ialloc, \rewards, \earnings)$ is a market equilibrium}\\
        & \leq \frac{(e_i + r_j) }{\MRC_i} \tag{$e'_i -r_j \leq e_i$}\\
        & = c_i(\dalloc_i) + c_{ij} \tag{$c_i(\dalloc_i) = e_i/ \MRC_i$ and $c_{ij} = r_j/ \MRC_i$}\\
        & \leq \frac{n^{1/p} \, c_i(\items)}{n} + c_{ij}. \tag{$\dalloc$ is $n^{1/p}$-$\PROP$}
     \end{align*}
     Hence, $n^{1/p}$-$\PROPone$.  
\end{proof}

\begin{algorithm}[t]
\textbf{Input:} Instance with $m$ indivisible chores, $n$ agents with costs $\{c_i\}_i$, and parameter $p \geq 1$.\\
\textbf{Output:} $\fPO$ and $n^{1/p}$-$\PROPone$ chore allocation $\ialloc = (\ialloc_1, \ialloc_2, \ldots, \ialloc_n)$.
\begin{itemize}
    \item Let $\dalloc$ be a fractional allocation that minimizes the $p$-mean of normalized disutilities, and let $(\dalloc,\rewards,\earnings)$ be the corresponding market equilibrium guaranteed to exist by~\Cref{theorem:div-chore-prop}.
    \item Let $\ialloc$ be the allocation obtained by rounding the market equilibrium $(\dalloc,\rewards, \earnings)$ using the algorithm of~\Cref{theorem:rounding-for-chores}.
    \item \textbf{return } $\ialloc = \left( \ialloc_1, \ialloc_2, \ldots, \ialloc_n \right)$
  \end{itemize}
  \caption{$\fPO$ and $\PROPone$ for indivisible chores} \label{algo:chores}
\end{algorithm}

\subsection{Indivisible Chores: Optimizing $p$-means}\label{sec:indivisible chores}

In this section, we consider directly optimizing the normalized $p$-mean of disutilities for indivisible chores. We start by proving two technical lemmas, which are crucially used in the proof of~\Cref{theorem:chore-EF1-result}. 

\begin{lemma} \label{lemma:indiv-chores-squeeze}
    Given any non-negative real numbers $a,b,\alpha, \beta > 0$, if $\max \{a, b\} \geq \max\{\alpha, \beta\}$ and $a^2 + b^2 > \alpha^2 + \beta^2$, then, for all $p > 2$, it holds that $a^p + b^p > \alpha^p + \beta^p$.
\end{lemma}
\begin{proof}   
    Without loss of generality, suppose that $a \geq b$ and $\alpha \geq \beta$. The given inequality, $\max \{a, b\} \geq \max\{\alpha, \beta\}$, implies that $a \geq \alpha$. If $b > \beta$, then the desired inequality easily follows (since $b > \beta$ and $a \geq \alpha$ imply that $a^p + b^p > \alpha^p + \beta^p$ for all $p > 2$). Thus, we assume that $\beta \geq b$.


    Define $x \coloneqq \beta/b$, $y \coloneqq \alpha/\beta$, and $z \coloneqq a/b$, and notice that $x, y, z \geq 1$. Substituting these in the desired inequality and rearranging give us
    \begin{align*}
        a^p + b^p > \alpha^p + \beta^p
        &\Longleftrightarrow b^p\left(1+\left(\frac{a}{b}\right)^p\right) > \beta^p\left(1+\left(\frac{\alpha}{\beta}\right)^p\right)  \\
        &\Longleftrightarrow 1+\left(\frac{a}{b}\right)^p > \left(\frac{\beta}{b}\right)^p\left(1+\left(\frac{\alpha}{\beta}\right)^p\right) \\
        &\Longleftrightarrow z > \left(x^p + (xy)^p - 1\right)^{\frac{1}{p}}.
    \end{align*}

    Define $f(p) \coloneqq \left(x^p + (xy)^p - 1\right)^{\frac{1}{p}}$.
    By rearranging the given inequality $a^2 + b^2 > \alpha^2 + \beta^2$, we have that $z > f(2)$, and our goal is to show that $z>f(p)$ for all $p>2$. To this end, it is sufficient to show that $f(p)$ is a non-increasing function of $p$, or equivalently that $f'(p) \leq 0$, for $p \geq 2$. 
    
    For notational convenience, let $u \coloneqq x$ and $v \coloneqq xy$, noting that $u,v \geq 1$. $f'(p)$ can then be expressed as
    \begin{align*}
        f'(p) &= (u^p + v^p - 1)^{\frac{1}{p}}\left(\frac{u^p \ln u + v^p \ln v}{p(u^p + v^p - 1)} - \frac{\ln (u^p + v^p - 1)}{p^2}\right) \\
        &= \frac{f(p)}{p^2(u^p + v^p - 1)}\left(p(u^p \ln u + v^p \ln v) - (u^p + v^p - 1) \ln (u^p + v^p - 1)\right). 
    \end{align*}
    Since $u, v \geq 1$ and $p \geq 2$, it holds that $u^p+v^p-1 \geq 0$, and therefore $f(p) \geq 0$ and $p^2(u^p + v^p - 1) \geq 0$. Thus, to show that $f'(p) \leq 0$ for all $p\geq 2$, it suffices to show that $p(u^p \ln u + v^p \ln v) - (u^p + v^p - 1) \allowbreak \ln (u^p + \allowbreak v^p - 1) \leq 0$, or, equivalently, that $g(p)\leq 0$ for all $p \geq 2$, where $g(p) \coloneqq p(u^p \ln u + v^p \ln v) - (u^p + v^p - 1) \ln (u^p + v^p - 1)$.
    Since $g(0) = 0$, it suffices to show that $g'(p) \leq 0$ for all $p \geq 0$.
    \begin{align*}
        g'(p) &= p(u^p (\ln u)^2 + v^p (\ln v)^2) - (u^p \ln u + v^p \ln v)\ln(u^p + v^p - 1) \\
        &= u^p\ln u(p\ln u - \ln(u^p + v^p - 1)) + v^p\ln v(p\ln v - \ln(u^p + v^p - 1)).
    \end{align*}
Since $u , v \geq 1$, we have that $p\ln u \leq \ln(u^p + v^p - 1)$ and $p\ln v \leq \ln(u^p + v^p - 1)$, while $u^p\ln u, v^p\ln v \geq 0$, and therefore $g'(p) \leq 0$ for all $p \geq 0$, which concludes the proof.  
\end{proof}

The proof of the following lemma has been deferred to Appendix~\ref{app: missing from chores}.
\begin{lemma}\label{lemma:indiv-chores-algebra}
    If $a$, $\alpha$, and $\beta$ are real numbers such that $a < \alpha < 1$ and $\beta < 1-2\alpha$, then the following two inequalities hold: (1) $\left(a + \frac{1-a}{1-\alpha}\cdot \beta \right)^2 + (1 - \alpha - \beta)^2 < a^2 + (1-\alpha)^2$, and (2) $a + \frac{1-a}{1-\alpha}\cdot \beta < 1 - \alpha$.
\end{lemma}

We now prove our main positive result for this section: for the case of two agents, the integral allocation that minimizes the $p$-mean of normalized disutilities is $\EFone$ and $\PO$, for all $p \geq 2$.

\begin{theorem} \label{theorem:chore-EF1-result}
For $n=2$ agents with additive costs over indivisible chores, and all $p \geq 2$, every allocation that minimizes the normalized $p$-mean of disutilities is $\EFone$ and $\PO$.
\end{theorem}

\begin{proof}   
    Towards a contradiction, assume that $\ialloc = (\ialloc_1, \ialloc_2)$ is an allocation that minimizes the $p$-mean of agents' normalized disutilities such that $\ialloc$ is not $\EFone$. Without loss of generality assume that agent $2$ envies agent $1$ even after the removal of any one chore. We show that starting from $\ialloc$, we can construct another allocation $\ialloc'$ that has a smaller $p$-mean, contradicting the optimality of $\ialloc$.
    
    Consider the case when $\tilde{c}_1(\ialloc_1) > \tilde{c}_2(\ialloc_1)$. On rearranging we get $1-\tilde{c}_2(\ialloc_1) > 1 - \tilde{c}_1(\ialloc_1)$, or equivalently, $\tilde{c}_2(\ialloc_2) > \tilde{c}_1(\ialloc_2)$, since by normalization we have $\tilde{c}_i(\ialloc_1) + \tilde{c}_i(\ialloc_2) = 1$. Together, $\tilde{c}_1(\ialloc_1) > \tilde{c}_2(\ialloc_1)$ and $\tilde{c}_2(\ialloc_2) > \tilde{c}_1(\ialloc_2)$ imply that the $p$-mean of the allocation $\ialloc' = (\ialloc_2, \ialloc_1)$ obtained by swapping the agents' bundles will be lower than the $p$-mean of $\ialloc = (\ialloc_1, \ialloc_2)$, a contradiction. 
    
    Henceforth, we consider the case when $\tilde{c}_1(\ialloc_1) \leq \tilde{c}_2(\ialloc_1)$. We prove that the allocation $\ialloc' = (\ialloc_1 \cup \{t^*\}, \ialloc_2 \setminus \{t^*\})$ obtained by transferring a specific chore $t^* \in \ialloc_2$ from agent $2$ to agent $1$ has a strictly lower $p$-mean. Specifically, the chore $t^*$ is such that $t^* \in arg \min_{t \in \ialloc_2} \ \frac{\tilde{c}_1(t)}{\tilde{c}_2(t)}$.
    By definition, we have that
    \begin{equation} \label{equation:indiv-chores-1}
        \frac{\tilde{c}_1(t^*)}{\tilde{c}_2(t^*)} \leq \frac{\tilde{c}_1(\ialloc_2) - \tilde{c}_1(t^*)}{\tilde{c}_2(\ialloc_2) - \tilde{c}_2(t^*)}.
    \end{equation}
    
    Let $\delta \geq 0$ be defined as per the following inequality:
    \begin{equation}\label{equation:definition-of-k}
        \frac{\tilde{c}_1(t^*) + \delta}{\tilde{c}_2(t^*)} = \frac{\tilde{c}_1(\ialloc_2) - \tilde{c}_1(t^*) - \delta}{\tilde{c}_2(\ialloc_2) - \tilde{c}_2(t^*)}.
    \end{equation}
    Note that,~\Cref{equation:indiv-chores-1} implies that such a $\delta \geq 0$ must exist. 
    Furthermore, we have the following: $1 = \tilde{c}_1(\ialloc_1) + \tilde{c}_1(\ialloc_2) = \tilde{c}_1(\ialloc_1) + (\tilde{c}_1(t^*) + \delta) + (\tilde{c}_1(\ialloc_2)-\tilde{c}_1(t^*) - \delta) = \tilde{c}_1(\ialloc_1) + (\tilde{c}_1(t^*) + \delta) + \frac{\tilde{c}_1(t^*) + \delta}{\tilde{c}_2(t^*)} \cdot (\tilde{c}_2(\ialloc_2) - \tilde{c}_2(t^*))$, which implies that 
    \begin{equation}\label{eq: delta useful fact}
        \frac{\tilde{c}_1(t^*) + \delta}{\tilde{c}_2(t^*)} = \frac{\tilde{c}_1(\ialloc_2)}{\tilde{c}_2(\ialloc_2)}.
    \end{equation}

In the remainder of the proof we show that $\left(\tilde{c}_1(\ialloc_1 \cup \{t^*\})\right)^2 + \tilde{c}_2(\ialloc_2 \setminus \{t^*\})^2 < \tilde{c}_1(\ialloc_1)^2 + \tilde{c}_2(\ialloc_2)^2$ and $\max\{\tilde{c}_1(\ialloc_1 \cup \{t^*\}), \tilde{c}_2(\ialloc_2 \setminus \{t^*\})\} < \max\{\tilde{c}_1(\ialloc_1), \tilde{c}_2(\ialloc_2)\}$, i.e., both the $2$-mean and the ``$\infty$-mean'' are lower in $\ialloc'$. Then, we can readily apply~\Cref{lemma:indiv-chores-squeeze} to obtain that $\left(\tilde{c}_1(\ialloc_1 \cup \{t^*\})\right)^p + \tilde{c}_2(\ialloc_2 \setminus \{t^*\})^p < \tilde{c}_1(\ialloc_1)^p + \tilde{c}_2(\ialloc_2)^p$ for all $p \geq 2$, contradicting the optimality of $\ialloc$, which would conclude the proof of~\Cref{theorem:chore-EF1-result}. We consider two cases: $\tilde{c}_1(\ialloc_1) < \tilde{c}_2(\ialloc_1)$ and $\tilde{c}_1(\ialloc_1) = \tilde{c}_2(\ialloc_1)$.
    
\paragraph*{Case 1. $\tilde{c}_1(\ialloc_1) < \tilde{c}_2(\ialloc_1)$:} 
Since $\ialloc$ is not $\EFone$ and agent 2 is envious, we have $c_2(\ialloc_1) < c_2(\ialloc_2 \setminus \{t^*\})$, or, equivalently, $\tilde{c}_2(\ialloc_1) < 1 - \tilde{c}_2(\ialloc_1) - \tilde{c}_2(t^*)$. Additionally, $\tilde{c}_2(\ialloc_1) < 1$ since agent $2$ envies agent $1$. Setting $a = \tilde{c}_1(\ialloc_1)$, $\alpha = \tilde{c}_2(\ialloc_1)$, and $\beta = \tilde{c}_2(t^*)$, both conditions $a<\alpha < 1$ and $\beta < 1-2\alpha$ of~\Cref{lemma:indiv-chores-algebra} are satisfied. Thus, by invoking~\Cref{lemma:indiv-chores-algebra}, Inequality~(1) implies that
\begin{align*}
	& \left(\tilde{c}_1(\ialloc_1) + \frac{1-\tilde{c}_1(\ialloc_1)}{1-\tilde{c}_2(\ialloc_1)} \tilde{c}_2(t^*)\right)^2 + (1 - \tilde{c}_2(\ialloc_1) - \tilde{c}_2(t^*))^2 < \tilde{c}_1(\ialloc_1)^2 + (1-\tilde{c}_2(\ialloc_1))^2 \\
 \implies & \left(\tilde{c}_1(\ialloc_1) + \frac{\tilde{c}_1(\ialloc_2))}{\tilde{c}_2(\ialloc_2)} \tilde{c}_2(t^*) \right)^2 + \tilde{c}_2(\ialloc_2 \setminus \{t^*\})^2 < \tilde{c}_1(\ialloc_1)^2 + \tilde{c}_2(\ialloc_2)^2 \\
	\implies & \left(\tilde{c}_1(\ialloc_1 \cup \{t^*\}) + \delta \right)^2 + \tilde{c}_2(\ialloc_2 \setminus \{t^*\})^2 < \tilde{c}_1(\ialloc_1)^2 + \tilde{c}_2(\ialloc_2)^2. \tag{\Cref{eq: delta useful fact}}
\end{align*}
This shows that the $2$-mean of normalized costs in allocation $\ialloc'$ is strictly less than that of $\ialloc$. 

We now show that $\max\{\tilde{c}_1(\ialloc'_1), \tilde{c}_2(\ialloc'_2)\} = \max\{\tilde{c}_1(\ialloc_1 \cup \{t^*\}), \tilde{c}_2(\ialloc_2 \setminus \{t^*\})\} < \max\{\tilde{c}_1(\ialloc_1), \tilde{c}_2(\ialloc_2)\}$. Since $\ialloc$ is not $\EFone$ and agent 2 is envious, we have $c_2(\ialloc_1) < c_2(\ialloc_2)$, and therefore $\tilde{c}_2(\ialloc_1) < \tilde{c}_2(\ialloc_2)$. Additionally, as per our assumption, $\tilde{c}_1(\ialloc_1) < \tilde{c}_2(\ialloc_1)$. Thus we have $\max\{\tilde{c}_1(\ialloc_1), \tilde{c}_2(\ialloc_2)\} = \tilde{c}_2(\ialloc_2)$. By invoking~\Cref{lemma:indiv-chores-algebra} as earlier, Inequality~(2) implies that $\tilde{c}_1(\ialloc_1) + \frac{1-\tilde{c}_1(\ialloc_1)}{1-\tilde{c}_2(\ialloc_1)} \tilde{c}_2(t^*) < 1-\tilde{c}_2(\ialloc_1)$. Via~\Cref{eq: delta useful fact}, the LHS of this inequality is equal to $\tilde{c}_1(\ialloc_1 \cup \{t^*\}) + \delta \geq \tilde{c}_1(\ialloc_1 \cup \{t^*\})$, while the RHS is equal to $\tilde{c}_2(\ialloc_2)$, overall implying that $\tilde{c}_1(\ialloc_1 \cup \{t^*\}) < \tilde{c}_2(\ialloc_2)$. 
By the definition of $t^*$ we have that $\tilde{c}_2(t^*)>0$, and therefore $\tilde{c}_2(\ialloc_2 \setminus \{t^*\}) < \tilde{c}_2(\ialloc_2)$. Therefore, we overall have $\max\{\tilde{c}_1(\ialloc_1 \cup \{t^*\}), \tilde{c}_2(\ialloc_2 \setminus \{t^*\})\} < \tilde{c}_2(\ialloc_2) = \max\{\tilde{c}_1(\ialloc_1), \tilde{c}_2(\ialloc_2)\}$. 


\paragraph*{Case 2. $\tilde{c}_1(\ialloc_1) = \tilde{c}_2(\ialloc_1)$:} 
Note that in this case,~\Cref{eq: delta useful fact} implies that $\tilde{c}_1(t^*) + \delta = \tilde{c}_2(t^*)$.
Furthermore, since $\ialloc$ is not $\EFone$ and agent 2 is envious, we have $c_2(\ialloc_2 \setminus \{t^*\}) > c_2(\ialloc_1)$, or, equivalently, $\tilde{c}_2(\ialloc_2 \setminus \{t^*\}) > \tilde{c}_2(\ialloc_1) = \tilde{c}_1(\ialloc_1)$, or, $\tilde{c}_2(\ialloc_2) > \tilde{c}_1(\ialloc_1 \cup \{t^*\})$.
We begin by showing that the $2$-mean of normalized disutilities in allocation $\ialloc'$ is strictly less than that of $\ialloc$:
\begin{align*}
    &\tilde{c}_1(\ialloc_1)^2 + \tilde{c}_2(\ialloc_2)^2 - \tilde{c}_1(\ialloc_1 \cup \{t^*\})^2 - \tilde{c}_2(\ialloc_2\setminus \{t^*\})^2\\
     & \geq \tilde{c}_1(\ialloc_1)^2 + (\tilde{c}_2(\ialloc_2 \setminus \{t^*\}) + \tilde{c}_2(t^*))^2 - (\tilde{c}_1(\ialloc_1) + (\tilde{c}_1(t^*)+\delta))^2 - \tilde{c}_2(\ialloc_2\setminus \{t^*\})^2 \tag{$\delta \geq 0$}\\
     & = \tilde{c}_1(\ialloc_1)^2 + (\tilde{c}_2(\ialloc_2 \setminus \{t^*\}) + \tilde{c}_2(t^*))^2 - (\tilde{c}_1(\ialloc_1) + \tilde{c}_2(t^*))^2 - \tilde{c}_2(\ialloc_2\setminus \{t^*\})^2 \tag{$\tilde{c}_1(t^*) + \delta = \tilde{c}_2(t^*)$}\\
     & = 2 \, \tilde{c}_2(\ialloc_2 \setminus \{t^*\}) \, \tilde{c}_2(t^*) - 2 \, \tilde{c}_1(\ialloc_1) \, \tilde{c}_2(t^*)\\
     & > 0. \tag{$\tilde{c}_2(\ialloc_2 \setminus \{t^*\}) > \tilde{c}_1(\ialloc_1)$}
\end{align*}

Next we show that $\max\{\tilde{c}_1(\ialloc_1 \cup \{t^*\}), \tilde{c}_2(\ialloc_2 \setminus \{t^*\})\} < \max\{\tilde{c}_1(\ialloc_1), \tilde{c}_2(\ialloc_2)\}$. Noticing that $\max\{\tilde{c}_1(\ialloc_1), \allowbreak \tilde{c}_2(\ialloc_2)\}\allowbreak = \max\{\tilde{c}_2(\ialloc_1), \tilde{c}_2(\ialloc_2)\} = \tilde{c}_2(\ialloc_2) > \tilde{c}_2(\ialloc_2 \setminus \{t^*\})$ we have
    \begin{align*}
    &\tilde{c}_2(\ialloc_2) - \tilde{c}_1(\ialloc_1 \cup \{t^*\})\\
    & \geq \tilde{c}_2(\ialloc_2) - \tilde{c}_1(\ialloc_1) - (\tilde{c}_1(t^*) + \delta) \tag{$\delta \geq 0$}\\
    & \geq \tilde{c}_2(\ialloc_2) - \tilde{c}_1(\ialloc_1) - \tilde{c}_2(t^*)    \tag{$\tilde{c}_1(t^*) + \delta = \tilde{c}_2(t^*)$}\\
    & > \tilde{c}_2(\ialloc_1) - \tilde{c}_1(\ialloc_1) \tag{$\tilde{c}_2(\ialloc_1) < \tilde{c}_2(\ialloc_2 \setminus \{t^*\})$}\\
    & = 0. \tag{$\tilde{c}_1(\ialloc_1) = \tilde{c}_2(\ialloc_2)$}
    \end{align*}
This concludes the proof of~\Cref{theorem:chore-EF1-result}.  
\end{proof}

Note that, using~\Cref{lemma:div-to-indiv} and~\Cref{theorem:chore-EF1-result} we obtain the following corollary which tightens the result of~\Cref{theorem:div-chore-prop} for the case of $n=2$ agents and divisible chores.
\begin{corollary}\label{corollary:divisible-chores-2-EF}
For $n=2$ agents with additive costs over divisible chores, and all $p \geq 2$, any allocation that minimizes the normalized $p$-mean of disutilities is envy-free.
\end{corollary}

We conclude by presenting some impossibility results for allocations that optimize the normalized $p$-means for indivisible chores. 

\begin{theorem}\label{theorem:indiv-chores-negative-wrt-p}
    For $n=2$ agents with additive costs over indivisible chores, and all $p < 2$, there exist instances such that every allocation that minimizes the normalized $p$-mean of disutilities is not $\PROPone$ (and hence not $\EFone$). 
\end{theorem}

\begin{proof}

    
We consider the following three cases based on the value of $p$.
    \paragraph{Case 1. $p < 1$.} Consider a instance with $n = 2$ agents and $m$ indivisible chores such that $c_{ij} = 1/m$ for all agents $i$ and chores $j$. Suppose that the $p$-mean minimizing allocation gives $x\in \{0,1,2, \ldots, m\}$ chores to agent 1 and $m-x$ chores to agent $2$. Since $x^p + (m-x)^p$ is concave in the interval $[0,m]$ for $p<1$, it is minimized (i.e., the $p$-mean $(\frac{1}{2}(x^p + (1-x)^p))^{1/p}$ is minimized), when $x$ is either $0$ or $m$, which corresponds to allocating all chores to either agent $1$ or $2$, violating $\PROPone$ (and $\EFone$).

    \paragraph{Case 2. $p = 1$.} Consider an instance with $n = 2$ agents and $m$ indivisible chores such that $c_{1j} = 1/m$ for all chores $j$, and $c_{2j} = 1/m+\epsilon$, for all $j < m$, and $c_{2m} = 1/m - (m-1)\epsilon$, for some small $\epsilon > 0$. The allocation which minimizes the normalized $p$-mean of disutilities allocates all chores except chore $m$ to agent $1$, which indeed violates $\PROPone$ (and $\EFone$).

    \paragraph{Case 3. $1 < p < 2$.} For this case, we will show that minimizing the $p$-mean of normalized disutilities for \emph{divisible} chores does not result in $\PROP$ allocations even for $n=2$ agents. Then, using~\Cref{lemma:div-to-indiv}, we can conclude that minimizing the normalized $p$-mean of disutilities for indivisible chores cannot result in $\PROPone$ (and hence $\EFone$) allocations.
    
    Consider an instance with $2$ agents and $2$ divisible chores such that $c_{11} = \frac{1}{2} + \delta$, $c_{12} = \frac{1}{2} - \delta$, $c_{21} = \frac{1}{2} + \varepsilon$ and $c_{22} = \frac{1}{2} - \varepsilon$. We will show that if $0 < \varepsilon < \delta < \varepsilon + (2\varepsilon + 1)(1 - p/2)$, then any allocation that minimizes the $p$-mean of normalized disutilities (for $p \in (1,2))$ is not $\PROP$.
    
    First, one can show that in such an instance, every allocation that minimizes the normalized $p$-mean allocates chore 2 entirely to agent 1; we prove this as~\Cref{lemma:div_chores_counterexample} in Appendix~\ref{app: missing from chores}. Now suppose that a normalized $p$-mean minimizing allocation allocates an $x \in [0,1]$ fraction of chore $1$ to agent 1 and $1-x$ of chore 1 to agent 2. Then, the $p$-mean of normalized disutilities can be written as $w_p(\dalloc) = \tilde{c}_2(\dalloc_2)^p + \tilde{c}_1(\dalloc_1)^p = \left(\frac{1}{2} + \varepsilon\right)^p(1-x)^p + \left(\frac{1}{2} - \delta + \left(\frac{1}{2} + \delta\right)x\right)^p$.
    In order to find the $x$ that minimizes the normalized $p$-mean, we compute the derivative with respect to $x$ and set it equal to $0$.
    \begin{align*}
        \frac{d w_p(x)}{d x} &= p\left(-\left(\frac{1}{2} + \varepsilon\right)^p(1-x)^{p-1} + \left(\frac{1}{2} - \delta + \left(\frac{1}{2} + \delta\right)x\right)^{p-1}\left(\frac{1}{2} + \delta\right)\right) = 0\\
        \implies & 0 = -\left(\frac{1}{2} + \varepsilon\right)^p(1-x)^{p-1} + \left(\frac{1}{2} - \delta + \left(\frac{1}{2} + \delta\right)x\right)^{p-1}\left(\frac{1}{2} + \delta\right)\\
        \implies & \left(\frac{1}{2} + \varepsilon\right)^p(1-x)^{p-1} = \left(\frac{1}{2} - \delta + \left(\frac{1}{2} + \delta\right)x \right)^{p-1}\left(\frac{1}{2} + \delta\right).
    \end{align*}
    Simplifying and rearranging, we get:
    \begin{align*}
        x = \frac{\left(\frac{1}{2} + \varepsilon\right)^{\frac{p}{p-1}} - \left(\frac{1}{2} - \delta\right)\left(\frac{1}{2}+\delta\right)^{\frac{1}{p-1}}}{\left(\frac{1}{2} + \delta\right)^{\frac{p}{p-1}} + \left(\frac{1}{2} + \varepsilon\right)^{\frac{p}{p-1}}}.
    \end{align*}    
    We wish to show that when $\delta < \varepsilon + (2\varepsilon + 1)(1 - p/2)$, the aforementioned optimal value of $x$ is strictly smaller than $\varepsilon/(1/2 + \varepsilon)$. This would imply that $\tilde{c}_2(\dalloc_2) = (1/2 + \varepsilon)(1-x) > 1/2$, i.e. the allocation is not $\PROP$. To this end, we simplify $x$ as follows, 
    \begin{align*}
        x &= \frac{\left(\frac{1}{2} + \varepsilon\right)^{\frac{p}{p-1}} - \left(\frac{1}{2} - \delta\right)\left(\frac{1}{2}+\delta\right)^{\frac{1}{p-1}}}{\left(\frac{1}{2} + \delta\right)^{\frac{p}{p-1}} + \left(\frac{1}{2} + \varepsilon\right)^{\frac{p}{p-1}}} \\
        &= \frac{\left(\frac{1}{2} + \varepsilon\right)^{\alpha + 1} - \left(\frac{1}{2} - \delta\right)\left(\frac{1}{2}+\delta\right)^{\alpha}}{\left(\frac{1}{2} + \delta\right)^{\alpha + 1} + \left(\frac{1}{2} + \varepsilon\right)^{\alpha + 1}} \tag{$\alpha \coloneqq \frac{1}{p-1}$} \\
        &= \frac{\left(\frac{1}{2} + \varepsilon\right) - \left(\frac{1}{2} - \delta\right)\left(1 + \Delta\right)^\alpha}{\left(\frac{1}{2} + \varepsilon\right) + \left(\frac{1}{2} + \delta\right)\left(1 + \Delta\right)^\alpha} \tag{$\Delta \coloneqq \frac{\delta - \varepsilon}{\frac{1}{2} + \varepsilon}$} \\
        &\leq \frac{\left(\frac{1}{2} + \varepsilon\right) - \left(\frac{1}{2} - \delta\right)\left(1 + \Delta\alpha\right)}{\left(\frac{1}{2} + \varepsilon\right) + \left(\frac{1}{2} + \delta\right)\left(1 + \Delta\alpha\right)} \tag{$(1 + \Delta)^\alpha \geq 1 + \Delta\alpha$, since $\alpha > 1$}\\
        &= \frac{\varepsilon + \delta - \Delta\alpha\left(\frac{1}{2} - \delta\right)}{1 + \varepsilon + \delta + \Delta\alpha\left(\frac{1}{2} + \delta\right)}. \tag{simplifying}
    \end{align*}
    
    Finally, to show that $\varepsilon/(1/2 + \varepsilon) > x$, we consider the  following difference,
    \begin{align*}
        &\frac{\varepsilon}{\frac{1}{2} + \varepsilon} - \frac{\varepsilon + \delta - \Delta\alpha\left(\frac{1}{2} - \delta\right)}{1 + \varepsilon + \delta + \Delta\alpha\left(\frac{1}{2} + \delta\right)} \\ &= \frac{\varepsilon}{\phi} - \frac{\varepsilon + \delta - \Delta\alpha\left(\frac{1}{2} - \delta\right)}{\psi} \tag{$\phi \coloneqq \frac{1}{2} + \varepsilon, \: \psi \coloneqq 1 + \varepsilon + \delta + \Delta\alpha\left(\frac{1}{2} + \delta\right)$} \\
        &= \frac{1}{\phi\psi}\left(\varepsilon + \Delta\alpha\varepsilon\left(\frac{1}{2} + \delta\right) + \Delta\alpha\varepsilon\left(\frac{1}{2} - \delta\right) - \frac{1}{2}\left(\varepsilon + \delta\right) + \frac{\Delta\alpha}{2}\left(\frac{1}{2} - \delta\right)\right) \\
        &=\frac{1}{\phi\psi}\left(\frac{\varepsilon - \delta}{2} + \varepsilon\Delta\alpha + \frac{\Delta\alpha}{2}\left(\frac{1}{2} - \delta\right)\right) \tag{simplifying} \\
        &= \frac{1}{\phi\psi}\left(\frac{-\Delta}{2}\left(\frac{1}{2} + \varepsilon\right) + \varepsilon\Delta\alpha + \frac{\Delta\alpha}{2}\left(\frac{1}{2} - \delta\right)\right) \tag{$\Delta = \frac{\delta - \varepsilon}{\frac{1}{2} + \varepsilon}$} \\
        &> \frac{\Delta}{2\phi\psi\alpha}\left(-\left(\frac{1}{2} + \varepsilon\right)\left(\frac{1}{\alpha}\right) + 2\varepsilon + \frac{1}{2} - \varepsilon - (2\varepsilon + 1)\left(1 - \frac{p}{2}\right)\right) \tag{$\delta < \varepsilon + (2\varepsilon + 1)\left(1 - \frac{p}{2}\right)$} \\
        &= \frac{\Delta}{2\phi\psi\alpha}\left(-\left(\frac{1}{2} + \varepsilon\right)(p-1) + 2\varepsilon + \frac{1}{2} - \varepsilon - (2\varepsilon + 1)\left(1 - \frac{p}{2}\right)\right) \tag{$\alpha = \frac{1}{p-1}$} \\
        &= \frac{\Delta}{2\phi\psi\alpha}\cdot 0 = 0. \tag{simplifying}
    \end{align*}
    Putting the two inequalities together, we get the desired inequality, completing the proof:
    \[
        x = \frac{\left(\frac{1}{2} + \varepsilon\right)^{\frac{p}{p-1}} - \left(\frac{1}{2} - \delta\right)\left(\frac{1}{2}+\delta\right)^{\frac{1}{p-1}}}{\left(\frac{1}{2} + \delta\right)^{\frac{p}{p-1}} + \left(\frac{1}{2} + \varepsilon\right)^{\frac{p}{p-1}}} \leq \frac{\varepsilon + \delta - \Delta\alpha\left(\frac{1}{2} - \delta\right)}{1 + \varepsilon + \delta + \Delta\alpha\left(\frac{1}{2} + \delta\right)} < \frac{\varepsilon}{\frac{1}{2} + \varepsilon}. 
    \]
  
\end{proof}

\begin{theorem}\label{thm: 14}
    For $n > 2$ agents with additive costs over indivisible chores, and all $p < n$, there exist instances such that every allocation that minimizes the normalized $p$-mean of disutilities is not $\PROPone$ (and hence not $\EFone$).
\end{theorem}
\begin{proof}    We proceed by analyzing the following three cases, based on the values of $p$.
    \paragraph{Case 1. $p < 1$.} Consider an instance with $n > 2$ agents and $n^2$ chores such that $c_{ij} = 1/n^2$ for all $i,j$. The minimum value of the normalized $p$ mean of disutilities will be $0$, and will be achieved when some agent $i$ gets an empty bundle. Due to the pigeonhole principle, one of remaining agent $j \in \agents \setminus \{i\}$, will be allocated a bundle having disutility at least $1/(n-1)$. Even after removing one chore from agent $j$'s bundle, its disutility will be at least $1/(n-1) - 1/n^2 > 1/n$. Hence, any allocation that minimizes the normalized $p$ mean of disutilities will not be $\PROPone$.
    \paragraph{Case 2. $p = 1$.} Here we analyze an instance with $n$ agents and $n^2$ chores. The disutilitites of the first agent are $c_{11} = 1/n^2 + (n-1)\varepsilon$ and $c_{1j} = 1/n^2 - \varepsilon$ for all $j \geq 2$, for a sufficiently small and positive $\varepsilon > 0$. The other agents are identical, $c_{ij} = 1/n^2$ for all $i \in \agents \setminus \{1\}$ and $j \in \items$. Note that any $p$ mean welfare minimizing allocation will allocate all chores in $\items \setminus \{1\}$ to agent $1$. Therefore, the disutility of agent $1$ will be at least $1 - 1/n^2 - (n-1)\varepsilon$. Even after the removal of one chore, the disutility of this agent will be strictly greater than $1/n$.
    
    \paragraph{Case 3. $p \in (1, n)$.} Towards a contradiction, we assume that every allocation that minimizes the normalized $p$ mean of disutilities, for $p \in (1,n)$, results in a $\PROPone$ allocation. Then, \Cref{lemma:div-to-indiv} implies that minimizing the normalized $p$ mean of disutilities for divisible chores results in $\PROP$ allocations. This directly contradicts \Cref{theorem:chores-div-neg-many-agents}, as the abovementioned optimization for divisible chores can only result in $\Omega(n^{1/p})$-$\PROP$ allocations.
\end{proof}


\section{Conclusion}\label{sec:Conclusion}
In this paper, we study the fundamental problem of computing fair and economically efficient allocations when the items to be allocated are all goods or all chores, all divisible or all indivisible. We define a new class of objective functions, the normalized $p$-mean objectives, and show that it produces fair and efficient allocations for each of the aforementioned settings. There are several natural directions for future research. Our results show that optimizing objectives that aren't necessarily \emph{welfare functions} can result in fair and efficient allocations; this motivates the exploration of fairness properties of allocation algorithms beyond welfarist rules. In particular, can we characterize the class of objective functions that result in fair allocations? For the case of $p$-mean objectives, are there axiomatic explanations why only $p \leq 0$ results in $\PROP$ allocations for divisible goods, or $p>2$ for indivisible chores results in $\EFone$ allocations for two agents? Can other objectives produce $\PO$ and $\EF1$ allocations of indivisible chores for $n>2$ agents? Extending our results to the case of arbitrary entitlements or mixed manna would also be interesting. Finally, studying the fairness properties of normalized $p$-mean objectives in other domains, like public good allocation and participatory budgeting, would also be interesting.

\bibliographystyle{alpha}
\bibliography{references}

@article{vickrey1961counterspeculation,
  title={Counterspeculation, auctions, and competitive sealed tenders},
  author={Vickrey, William},
  journal={The Journal of finance},
  volume={16},
  number={1},
  pages={8--37},
  year={1961},
  publisher={JSTOR}
}

@article{clarke1971multipart,
  title={Multipart pricing of public goods},
  author={Clarke, Edward H},
  journal={Public choice},
  pages={17--33},
  year={1971},
  publisher={JSTOR}
}

@article{alon2010strategyproof,
  title={Strategyproof approximation of the minimax on networks},
  author={Alon, Noga and Feldman, Michal and Procaccia, Ariel D and Tennenholtz, Moshe},
  journal={Mathematics of Operations Research},
  volume={35},
  number={3},
  pages={513--526},
  year={2010},
  publisher={INFORMS}
}

@inproceedings{viswanathan2023general,
  title={A general framework for fair allocation under matroid rank valuations},
  author={Viswanathan, Vignesh and Zick, Yair},
  booktitle={Proceedings of the 24th ACM Conference on Economics and Computation},
  pages={1129--1152},
  year={2023}
}

@inproceedings{barman2022universal,
  title={Universal and tight online algorithms for generalized-mean welfare},
  author={Barman, Siddharth and Khan, Arindam and Maiti, Arnab},
  booktitle={Proceedings of the AAAI Conference on Artificial Intelligence},
  volume={36},
  pages={4793--4800},
  year={2022}
}

@inproceedings{garg2023new,
  title={New algorithms for the fair and efficient allocation of indivisible chores},
  author={Garg, Jugal and Murhekar, Aniket and Qin, John},
  booktitle={Proceedings of the Thirty-Second International Joint Conference on Artificial Intelligence, IJCAI-23},
  pages={2710--2718},
  year={2023}
}

@article{aziz2020justifications,
  title={Justifications of welfare guarantees under normalized utilities},
  author={Aziz, Haris},
  journal={ACM SIGecom Exchanges},
  volume={17},
  number={2},
  pages={71--75},
  year={2020},
  publisher={ACM New York, NY, USA}
}

@article{plaut2020almost,
  title={Almost envy-freeness with general valuations},
  author={Plaut, Benjamin and Roughgarden, Tim},
  journal={SIAM Journal on Discrete Mathematics},
  volume={34},
  number={2},
  pages={1039--1068},
  year={2020},
  publisher={SIAM}
}

@inproceedings{barman2020tight,
  title={Tight Approximation Algorithms for p-Mean Welfare Under Subadditive Valuations},
  author={Barman, Siddharth and Bhaskar, Umang and Krishna, Anand and Sundaram, Ranjani G},
  booktitle={28th Annual European Symposium on Algorithms (ESA 2020)},
  volume={173},
  pages={11},
  year={2020},
  organization={Schloss Dagstuhl--Leibniz-Zentrum f $\{$$\backslash$" u$\}$ r Informatik}
}

@inproceedings{barman2021uniform,
  title={Uniform welfare guarantees under identical subadditive valuations},
  author={Barman, Siddharth and Sundaram, Ranjani G},
  booktitle={Proceedings of the Twenty-Ninth International Conference on International Joint Conferences on Artificial Intelligence},
  pages={46--52},
  year={2021}
}

@inproceedings{chaudhury2021fair,
  title={Fair and efficient allocations under subadditive valuations},
  author={Chaudhury, Bhaskar Ray and Garg, Jugal and Mehta, Ruta},
  booktitle={Proceedings of the AAAI Conference on Artificial Intelligence},
  volume={35},
  pages={5269--5276},
  year={2021}
}

@inproceedings{garg2022tractable,
  title={Tractable Fragments of the Maximum Nash Welfare Problem},
  author={Garg, Jugal and Husic, Edin and Murhekar, Aniket},
  booktitle={Web and Internet Economics: 18th International Conference, WINE 2022, Troy, NY, USA, December 12--15, 2022, Proceedings},
  volume={13778},
  pages={362},
  year={2022},
  organization={Springer Nature}
}

@inproceedings{barman2023fair,
  title={Fair Chore Division under Binary Supermodular Costs},
  author={Barman, Siddharth and Narayan, Vishnu and Verma, Paritosh},
  booktitle={Proceedings of the 2023 International Conference on Autonomous Agents and Multiagent Systems},
  pages={2863--2865},
  year={2023}
}

@inproceedings{ebadian2022fairly,
  title={How to Fairly Allocate Easy and Difficult Chores},
  author={Ebadian, Soroush and Peters, Dominik and Shah, Nisarg},
  booktitle={Proceedings of the 21st International Conference on Autonomous Agents and Multiagent Systems},
  pages={372--380},
  year={2022}
}

@inproceedings{garg2022fair,
  title={Fair and efficient allocations of chores under bivalued preferences},
  author={Garg, Jugal and Murhekar, Aniket and Qin, John},
  booktitle={Proceedings of the AAAI Conference on Artificial Intelligence},
  volume={36},
  pages={5043--5050},
  year={2022}
}

@book{brandt2016handbook,
  title={Handbook of computational social choice},
  author={Brandt, Felix and Conitzer, Vincent and Endriss, Ulle and Lang, J{\'e}r{\^o}me and Procaccia, Ariel D},
  year={2016},
  publisher={Cambridge University Press}
}

@incollection{graham1979optimization,
  title={Optimization and approximation in deterministic sequencing and scheduling: a survey},
  author={Graham, Ronald Lewis and Lawler, Eugene Leighton and Lenstra, Jan Karel and Kan, AHG Rinnooy},
  booktitle={Annals of discrete mathematics},
  volume={5},
  pages={287--326},
  year={1979},
  publisher={Elsevier}
}

@article{abdulkadiroglu2013matching,
  title={Matching markets: Theory and practice},
  author={Abdulkadiroglu, Atila and S{\"o}nmez, Tayfun},
  journal={Advances in Economics and Econometrics},
  volume={1},
  pages={3--47},
  year={2013},
  publisher={Cambridge Univ. Press}
}

@inproceedings{barman2022truthful,
  title={Truthful and fair mechanisms for matroid-rank valuations},
  author={Barman, Siddharth and Verma, Paritosh},
  booktitle={Proceedings of the AAAI Conference on Artificial Intelligence},
  volume={36},
  number={5},
  pages={4801--4808},
  year={2022}
}

@inproceedings{barman2019proximity,
  title={On the proximity of markets with integral equilibria},
  author={Barman, Siddharth and Krishnamurthy, Sanath Kumar},
  booktitle={Proceedings of the AAAI Conference on Artificial Intelligence},
  volume={33},
  pages={1748--1755},
  year={2019}
}

@article{yuen2023extending,
  title={Extending the characterization of maximum Nash welfare},
  author={Yuen, Sheung Man and Suksompong, Warut},
  journal={Economics Letters},
  volume={224},
  pages={111030},
  year={2023},
  publisher={Elsevier}
}

@article{eisenberg1959consensus,
  title={Consensus of subjective probabilities: The pari-mutuel method},
  author={Eisenberg, Edmund and Gale, David},
  journal={The Annals of Mathematical Statistics},
  volume={30},
  number={1},
  pages={165--168},
  year={1959},
  publisher={JSTOR}
}

@article{nash1950bargaining,
  title={The bargaining problem},
  author={Nash Jr, John F},
  journal={Econometrica: Journal of the econometric society},
  pages={155--162},
  year={1950},
  publisher={JSTOR}
}

@book{arrow1981handbook,
  title={Handbook of mathematical economics},
  author={Arrow, Kenneth Joseph and Intriligator, Michael D and Hildenbrand, Werner and Sonnenschein, Hugo},
  volume={1},
  year={1981},
  publisher={North-Holland Amsterdam}
}

@article{varian1974equity,
  title={Equity, envy, and efficiency},
  author={Varian, Hal R},
  journal={Journal of Economic Theory},
  volume={9},
  number={1},
  pages={63--91},
  year={1974},
  publisher={Elsevier BV}
}

@article{bhaskar2021approximate,
  title={On Approximate Envy-Freeness for Indivisible Chores and Mixed Resources},
  author={Bhaskar, Umang and Sricharan, AR and Vaish, Rohit},
  journal={Approximation, Randomization, and Combinatorial Optimization. Algorithms and Techniques},
  year={2021}
}

@inproceedings{lipton2004approximately,
  title={On approximately fair allocations of indivisible goods},
  author={Lipton, Richard J and Markakis, Evangelos and Mossel, Elchanan and Saberi, Amin},
  booktitle={Proceedings of the 5th ACM Conference on Electronic Commerce},
  pages={125--131},
  year={2004}
}

@inproceedings{conitzer2017fair,
  title={Fair public decision making},
  author={Conitzer, Vincent and Freeman, Rupert and Shah, Nisarg},
  booktitle={Proceedings of the 2017 ACM Conference on Economics and Computation},
  pages={629--646},
  year={2017}
}

@article{steinhaus1948problem,
  title={The problem of fair division},
  author={Steinhaus, Hugo},
  journal={Econometrica},
  volume={16},
  pages={101--104},
  year={1948}
}

@book{foley1966resource,
  title={Resource allocation and the public sector},
  author={Foley, Duncan Karl},
  year={1966},
  publisher={Yale University}
}

@article{branzei2023algorithms,
  title={Algorithms for competitive division of chores},
  author={Br{\^a}nzei, Simina and Sandomirskiy, Fedor},
  journal={Mathematics of Operations Research},
  year={2023},
  publisher={INFORMS}
}

@article{bogomolnaia2017competitive,
  title={Competitive division of a mixed manna},
  author={Bogomolnaia, Anna and Moulin, Herv{\'e} and Sandomirskiy, Fedor and Yanovskaya, Elena},
  journal={Econometrica},
  volume={85},
  number={6},
  pages={1847--1871},
  year={2017},
  publisher={Wiley Online Library}
}

@book{mas1995microeconomic,
  title={Microeconomic theory},
  author={Mas-Colell, Andreu and Whinston, Michael Dennis and Green, Jerry R and others},
  volume={1},
  year={1995},
  publisher={Oxford university press New York}
}

@article{caragiannis2019unreasonable,
  title={The unreasonable fairness of maximum Nash welfare},
  author={Caragiannis, Ioannis and Kurokawa, David and Moulin, Herv{\'e} and Procaccia, Ariel D and Shah, Nisarg and Wang, Junxing},
  journal={ACM Transactions on Economics and Computation (TEAC)},
  volume={7},
  number={3},
  pages={1--32},
  year={2019},
  publisher={ACM New York, NY, USA}
}

\appendix

\section{Additional Market Preliminaries}\label{appendix:market-prelims}

\paragraph{Market equilibrium for chores.}
Analogous to divisible goods, we can define a Fisher market equilibrium for chores. Given agents $\agents$ and chores $\items$, we say that the triple $(\dalloc, \rewards, \earnings)$, where $r_j$ is the reward/payment for doing one unit of chore $j$ and $e_i$ is the earning goal of agent $i$, is a \emph{market equilibrium} iff the following three conditions are satisfied.

\begin{itemize}
    \item {\bf Market clears:} For every chore $j$, either $p_j=0$ or $\sum_{i=1} x_{ij}=1$.
    \item {\bf Maximum reward-per-cost allocation:} If $x_{ij}>0$, then either $c_{ij}=0$ or $\frac{r_j}{c_{ij}} \geq \frac{r_\ell}{c_{i\ell}}$ for every chore $\ell$ with $c_{i \ell}>0$.
    \item {\bf Earning-goals met:} The total earning of each agent is exactly $e_i$, i.e., $e_i = \sum_{j} x_{ij}r_j$.
\end{itemize}

Overloading notations, we will use $\MRC_i = \{j \in \items : r_j/c_{ij} \geq r_\ell/c_{i\ell} \text{ for all } \ell \in \items\}$ to denote all the maximum reward-per-cost chores for agent $i$ in equilibrium $(\dalloc, \rewards, \earnings)$ as well as the value $\MRC_i = \max_j r_j/c_{ij}$. 



\section{Proofs missing from Section~\ref{SEC:GOODS}}\label{app:missing proofs goods}


\begin{proof}{Proof of~\Cref{lemma:div-goods-not-EF}}
    Consider an instance with $n=3$ agents and two divisible goods such that $v_{11} = 0$ and $v_{12} = 1$; $v_{21} = 0.5$ and $v_{22} = 0.5$; $v_{31} = 0.8$ and $v_{32} = 0.2$. We will show that any allocation that maximizes the normalized $p$-mean for $p \rightarrow - \infty$ is not $\EF$. Note that by continuity, this implies that the same holds for a sufficiently large but finite $p$.
    
    An allocation $\dalloc^*$ that maximizes the normalized $p$-mean for $p \rightarrow - \infty$, (i.e., maximizes the minimum normalized utility) can be computed by solving the following linear program: $\max \ k$ such that $\sum_{j=1}^2 \tilde{v}_{ij} \cdot x^*_{ij} \geq k$ for each $i \in \{1,2,3\}$ and $\sum_{i=1}^3 x^*_{ij} \leq 1$ for each $j \in \{1,2\}$. It can easily be verified that the optimum of this program occurs when item $1$ is divided among agent $2$ and $3$ such that $x^*_{21} = 7/17$, $x^*_{31} = 10/17$ and item $2$ is divided such that $x^*_{12} = 8/17$ and $x^*_{22   } = 9/17$. However, this allocation is not $\EF$ since agent $1$ envies $2$, i.e., $v_{1}(\dalloc^*_1) = 8/17 < 9/17 = v_1(\dalloc^*_2)$. This completes the proof. Note that, this counter-example can be easily extended to $n>3$ agents by simply adding pairs of new agents and goods such the agents in these pairs only like the corresponding good.  
\end{proof}

\begin{lemma} \label{lemma:div_goods_counterexample}
    In any instance with $m=2$ divisible goods and $n=2$ agents such that $v_{11} = 1/2 + \varepsilon$, $v_{12} = 1/2 - \varepsilon$, $v_{21} = 1/2 + \delta$, and $v_{22} = 1/2 - \delta$ with $\delta > \epsilon$, for all $p\in(0,1)$, every allocation that maximizes the normalized $p$-mean of utilities must allocate good $2$ entirely to agent $1$.    
\end{lemma}
\begin{proof}   
    First, consider an allocation in which every agent has a nonzero share of every good. We will show that there exists $t_1, t_2 \in (0,1)$ such that transferring $t_1$ fraction of good 1 to agent 2 and $t_2$ fraction of good 2 to agent 1 results in a strictly larger $p$-mean. To achieve this, we can pick $t_1$ and $t_2$ so that the utility of agent 2 does not change after the transfer, i.e.,
    \begin{align}\label{equation:div-2-agent-item-lemma-eq1}
         t_1(1/2 + \delta) - t_2(1/2 - \delta) = \frac{1}{2}(t_1 - t_2) + \delta(t_1 + t_2) = 0.
    \end{align}
    For such a choice of $t_1$ and $t_2$, we can show that agent $1$'s utility always increases because the change in agent $1$'s utility can be written as,
    \begin{align*}
        t_2(1/2 - \varepsilon) - t_1(1/2 + \varepsilon) &= \frac{1}{2}(t_2 - t_1) - \varepsilon(t_1 + t_2) \\
        &= \delta(t_1 + t_2) - \varepsilon(t_1 + t_2) \tag{via~\Cref{equation:div-2-agent-item-lemma-eq1}} \\
        &> 0. \tag{$\delta > \varepsilon$}
    \end{align*}

    Therefore, either good $2$ is fully allocated to agent $1$ and we are done, or, good 1 is fully allocated to agent 2. Suppose the latter occurs, and suppose that $x$ is the fraction of good $2$ allocated to agent $1$. Then, $\tilde{v}_1(\dalloc_1) = (1/2 - \varepsilon)x$ and $\tilde{v}_2(\dalloc_2) = 1/2 + \delta + (1-x)(1/2 - \delta) = 1 - x(1/2 - \delta)$. We show that transferring the entirety of good 2 to agent 1 results in a larger $p$-mean.  The normalized $p$-mean of utilities is equal to
    \begin{align*}
        w_p(x) = \left(x\left(\frac{1}{2}-\varepsilon\right)\right)^p + \left(1-x\left(\frac{1}{2}-\delta\right)\right)^p.
    \end{align*}
    We first consider the derivative with respect to $x$ and show that it is positive.
    \begin{align*} 
        \frac{\partial w_p(x)}{\partial x} &= p\left(\left(\frac{1}{2} - \varepsilon\right)^px^{p-1} - \left(1 - x\left(\frac{1}{2} - \delta\right)\right)^{p-1}\left(\frac{1}{2} - \delta\right)\right) > 0 \\
        &\Longleftrightarrow \left(\frac{1}{2} - \varepsilon\right)^px^{p-1} > \left(1 - x\left(\frac{1}{2} - \delta\right)\right)^{p-1}\left(\frac{1}{2} - \delta\right) \\
        &\Longleftrightarrow x < \frac{\left(\frac{1}{2} - \delta\right)^{\frac{1}{p-1}}}{\left(\frac{1}{2} - \varepsilon\right)^{\frac{p}{p-1}} + \left(\frac{1}{2} - \delta\right)^{\frac{p}{p-1}}}. \tag{$0<p<1$}
    \end{align*}
    Showing that the stated upper bound on $x$ is larger than $1$ would imply that the derivative is positive everywhere in the interval $[0,1]$, and therefore the fraction of good $2$ allocated to agent $1$ should be equal to $1$.
    \begin{align*}
        \frac{\left(\frac{1}{2} - \delta\right)^{\frac{1}{p-1}}}{\left(\frac{1}{2} - \varepsilon\right)^{\frac{p}{p-1}} + \left(\frac{1}{2} - \delta\right)^{\frac{p}{p-1}}} \tag{$\delta > \varepsilon, \: 0 < p < 1$} &> \frac{\left(\frac{1}{2} - \delta\right)^{\frac{1}{p-1}}}{2\left(\frac{1}{2} - \delta\right)^{\frac{p}{p-1}}}\\ &= \frac{1}{1-2\delta} > 1.
    \end{align*}
    The lemma stands proved.  
\end{proof}

\section{Proofs missing from Section~\ref{SEC:CHORES}}\label{app: missing from chores}

\begin{proof}{Proof of~\Cref{theorem:div-chore-prop} continued}

Here, we prove that the allocations that minimize the normalized $p$-mean objective for divisible chores admit a market interpretation.
For the remainder of the proof we assume without loss of generality that $c_{ij} > 0$ for all agents $i$ and chores $j$; if $c_{ij} = 0$ for some agent $i$ and chore $j$, every $p$-mean optimal allocation will allocate chore $j$ to agent $i$.

We analyze the convex program whose objective is to minimize the $p$-mean of agents' normalized disutilities for $p>1$.  Minimizing $\left( \sum_{i=1}^n \frac{1}{n} \left( \sum_{j=1}^m \tilde{c}_{ij}x_{ij} \right)^{p}\right)^{1/p}$ is equivalent to minimizing $\sum_{i=1}^n \left( \sum_{j=1}^m \tilde{c}_{ij}x_{ij} \right)^{p}$, this convex program (indeed, the objective is convex, since each term of the objective $\left( \sum_{j=1}^m\tilde{c}_{ij}x_{ij} \right)^{p}$ is convex) can be written as follows:
\begin{align*}
    \text{minimize} \ \ \sum_{i=1}^n & \left( \sum_{j=1}^m \tilde{c}_{ij}x_{ij} \right)^{p}\\
    \text{subject to: }~~ -\sum_{i=1}^n x_{ij} & \leq -1 \ \ \text{, for all } j \in \items\\
    -x_{ij} & \leq 0 \ \ \text{, for all } i \in \agents, j \in \items
\end{align*}
    \paragraph{Analysis of the Convex Program.} Let $\dalloc^*$ be the optimal solution to this program, and let $r^*_j \geq 0$ and $\kappa^*_{ij} \geq 0$ be the associated optimal dual variables corresponding to the constraints $-\sum_{i=1}^n x_{ij} \leq -1$ and $-x_{ij} \leq 0$, respectively. The KKT conditions imply that

    \begin{align*}
        r^*_{j} \left( - \sum_{i=1}^n x^*_{ij} + 1 \right) & = 0, \text{ for all }  j \in \items \label{equation:kkt-1 chore} \numberthis\\
        \kappa^*_{ij}x^*_{ij} & = 0, \text{ for all } i \in \agents, j \in \items \label{equation:kkt-2 chore} \numberthis\\
         \frac{\partial}{\partial x_{ij}}  \sum_{i=1}^n  \left( \sum_{j=1}^m \tilde{c}_{ij}x_{ij} \right)^{p} \Biggr|_{x^*_{ij}} & + r^*_j \,  \frac{\partial}{\partial x_{ij}} \left(-\sum_{i=1}^n x_{ij} +1\right) \Biggr|_{x^*_{ij}} + \kappa^*_{ij} \,  \frac{\partial}{\partial x_{ij}} (-x_{ij}) \Biggr|_{x^*_{ij}} = 0 ,\forall i \in \agents, j \in \items. \numberthis
    \end{align*}

Note that, letting $\tilde{c}_i(\dalloc^*_i) = \sum_{j=1}^m \tilde{c}_{ij}x^*_{ij}$, stationarity (the third condition) implies:
\begin{equation}
p \tilde{c}_{ij} \, \tilde{c}_i(\dalloc^*_i)^{p-1} = r^*_j + \kappa^*_{ij} \text{ for all } i \in \agents, j \in \items. \label{equation:kkt-3 chore} 
\end{equation}

    Since $\tilde{c}_{ij}>0$, then~\Cref{equation:kkt-3 chore} implies that $p \, \tilde{c}_i(\dalloc^*_i)^{p-1} = \frac{r^*_j + \kappa^*_{ij}}{\tilde{c}_{ij}}$. Additionally, if $x^*_{ij} > 0$,~\Cref{equation:kkt-2 chore} implies that $\kappa^*_{ij} = 0$, and therefore, $p \, \tilde{c}_i(\dalloc^*_i)^{p-1} = \frac{r^*_j}{\tilde{c}_{ij}}$. Hence, if $x^*_{ij} > 0$, for any $\ell \in \items$ we have
    \begin{equation}
        \frac{r^*_j}{\tilde{c}_{ij}} = p \, \tilde{c}_i(\dalloc^*_i)^{p-1} = \frac{r^*_{\ell} + \kappa^*_{i\ell}}{\tilde{c}_{i\ell}} \geq \frac{r^*_{\ell}}{\tilde{c}_{i\ell}}. \label{equation:kkt-4 chore}
    \end{equation}

    \paragraph{Market Interpretation.} 
    Let $e^*_i \coloneqq p \, \tilde{c}_i(\dalloc_i^*)^{p}$. We will show that $(\dalloc^*, \rewards^*, \earnings^*)$ is a market equilibrium. 
    
    First, notice that $r^*_j > 0$ for all $j \in \items$:
    for $i \in \agents$ such that $x^*_{ij} > 0$, we have $r^*_j = p \tilde{c}_{ij} \tilde{c}_i(\dalloc^*_i)^{p-1}$, where $\tilde{c}_{ij} > 0$ and and $\tilde{c}_i(\dalloc^*_i)^{p-1} > 0$ (since $x^*_{ij} > 0$). Then, note that the \emph{market clears} (i.e., $\sum_{i=1}^n x^*_{ij} = 1$, for all $j \in \items$): this is a direct implication of~\Cref{equation:kkt-1 chore} and the fact that $r^*_j > 0$ for all $j \in \items$. 
    
    Second, $\dalloc^*$ is a \emph{maximum reward-per-cost} allocation, i.e. if $x^*_{ij}>0$, then~\Cref{equation:kkt-4 chore} implies that $\MRC_i \coloneqq  \frac{r^*_j}{\tilde{c}_{ij}} \geq \frac{r^*_\ell}{\tilde{c}_{i\ell}}$ for all $\ell \in \items$. Finally, every agent $i \in \agents$ meets their earning goal $e^*_i$. The amount earned by agent $i$ is
    \begin{align*}
        \sum_{j:\ x^*_{ij}>0} x^*_{ij} r^*_j &= \sum_{j:\ x^*_{ij}>0} x^*_{ij} \tilde{c}_{ij} \cdot \left( \frac{r^*_j}{\tilde{c}_{ij}} \right) \\
        &= \sum_{j:\ x^*_{ij}>0} x^*_{ij} \tilde{c}_{ij} \cdot \left( p \, \tilde{c}_i(\dalloc^*_i)^{p-1} \right) \tag{\Cref{equation:kkt-4 chore}}\\
        &= p \, \tilde{c}_i(\dalloc^*_i)^{p-1} \sum_{j:\ x^*_{ij}>0} x^*_{ij} \tilde{c}_{ij} \\
        &=  p \, \tilde{c}_i(\dalloc^*_i)^{p-1} \tilde{c}_i(\dalloc^*_i) = e^*_i.    \end{align*}
    This completes the proof.  
\end{proof}


\begin{proof}{Proof of~\Cref{lemma:indiv-chores-algebra}}
    To prove the first inequality, we consider the following difference,
    \begin{align*}
        &\left(a + \frac{\beta(1-a)}{1-\alpha}\right)^2 + (1 - \alpha - \beta)^2 - (a^2 + (1-\alpha)^2)\\
        & = a^2 + \frac{\beta^2(1-a)^2}{(1-\alpha)^2} + \frac{2a\beta(1-a)}{1-\alpha} + (1 - \alpha)^2 + \beta^2 - 2\beta(1-\alpha) - a^2 - (1-\alpha)^2\\
        &=\beta\left(\frac{\beta(1-a)^2}{(1-\alpha)^2} + \frac{2a(1-a)}{1-\alpha} + \beta - 2(1-\alpha)\right) \\
        &< \beta\left(\frac{(1-2\alpha)(1-a)^2}{(1-\alpha)^2} + \frac{2a(1-a)}{1-\alpha} + (1-2\alpha) - 2(1-\alpha)\right) \tag{$\beta < 1 - 2\alpha$} \\
        &= \beta\left(\frac{(1-2\alpha)(1-a)^2}{(1-\alpha)^2} + \frac{2a(1-a)}{1-\alpha} - 1\right) \\
        &= \frac{\beta}{(1-\alpha)^2}((1-2\alpha)(1-a)^2 + 2a(1-a)(1-\alpha) - (1-\alpha)^2)\\
        &= \frac{\beta}{(1-\alpha)^2}(2a\alpha - a^2 - \alpha^2) \tag{simplifying} \\
        &= -\frac{\beta}{(1-\alpha)^2}(a - \alpha)^2 < 0 \tag{$a < \alpha$}. 
    \end{align*}
    This establishes the first inequality. For the second inequality, we have
    \[
         a + \frac{1-a}{1-\alpha}\cdot \beta < 1 - \alpha \Longleftarrow \  \alpha < \frac{1-a}{1-\alpha}(1-\alpha-\beta).
    \]
    The last inequality above is true because $a < \alpha < 1$ (i.e., $1 < \frac{1-a}{1-\alpha}$ ) and $\alpha < 1 - \alpha - \beta$.  
\end{proof}

\begin{lemma} \label{lemma:div_chores_counterexample}
    For the divisible chores allocation instance below and $p \in (1,2)$, the allocation that minimizes the $p$-mean of normalized disutilities is such that, when $\delta > \varepsilon$, chore 2 is fully allocated to agent 1.
    $$\begin{pmatrix}
        \frac{1}{2} + \delta \ \ & ~~\frac{1}{2} - \delta\\
        \frac{1}{2} + \varepsilon \ \ & ~~\frac{1}{2} - \varepsilon
    \end{pmatrix}$$
\end{lemma}
\begin{proof}   
    First, consider an allocation in which each agent has a nonzero fraction of each chore. Then, there exists $t_1$ and $t_2$ such that transferring a $t_1$ fraction of chore 1 to agent 2 and $t_2$ fraction of chore 2 to agent 1 results in a strictly smaller $p$-mean. To achieve this, pick $t_1$ and $t_2$ so that the cost of agent 2 does not change, i.e. $t_1(1/2 + \varepsilon) - t_2(1/2 - \varepsilon) = \frac{1}{2}(t_1 - t_2) + \varepsilon(t_1 + t_2) = 0$.
    Then, consider the change in cost for agent 1.
    \begin{align*}
        t_2(1/2 - \delta) - t_1(1/2 + \delta) &= \frac{1}{2}(t_2 - t_1) - \delta(t_1 + t_2) \\
        &= \varepsilon(t_1 + t_2) - \delta(t_1 + t_2) \tag{$\frac{1}{2}(t_1 - t_2) + \varepsilon(t_1 + t_2) = 0$} \\
        &< 0. \tag{$\delta > \varepsilon$}
    \end{align*}
    Since performing this transfer is giving a Pareto improvement, the allocation could not have been $p$-mean optimal. If agent $1$ receives all of chore $2$ we are done. We show that if chore 1 is fully allocated to agent 2, then transferring the entirety of chore 2 to agent 1 results in a smaller $p$-mean. Let $x$ be the fraction of chore 2 allocated to agent 1. Then, $\tilde{c}_1(\dalloc_1) = (1/2 - \delta)x$ and $\tilde{c}_2(\dalloc_2) = 1/2 + \varepsilon + (1-x)(1/2 - \varepsilon) = 1 - x(1/2 - \varepsilon)$. We then write the normalized $p$-mean as follows:
    \begin{align*}
        w_p(x) = \left(x\left(\frac{1}{2}-\delta\right)\right)^p + \left(1-x\left(\frac{1}{2}-\varepsilon\right)\right)^p.
    \end{align*}
    We then consider the derivative with respect to $x$ and show that it is negative (thus larger $x$ implies smaller $p$-mean).
    \begin{align*} 
        \frac{\partial w_p(x)}{\partial x} &= p\left(\left(\frac{1}{2} - \delta\right)^px^{p-1} - \left(1 - x\left(\frac{1}{2} - \varepsilon\right)\right)^{p-1}\left(\frac{1}{2} - \varepsilon\right)\right) < 0 \\
        &\Longleftrightarrow \left(\frac{1}{2} - \delta\right)^px^{p-1} < \left(1 - x\left(\frac{1}{2} - \varepsilon\right)\right)^{p-1}\left(\frac{1}{2} - \varepsilon\right) \\
        &\Longleftrightarrow x < \frac{\left(\frac{1}{2} - \varepsilon\right)^{\frac{1}{p-1}}}{\left(\frac{1}{2} - \delta\right)^{\frac{p}{p-1}} + \left(\frac{1}{2} - \varepsilon\right)^{\frac{p}{p-1}}}.
    \end{align*}
    Showing that the stated upper bound on $x$ is larger than $1$ would imply that the derivative is negative everywhere in the interval $[0,1]$, and therefore the fraction of chore $2$ allocated to agent $1$ should be equal to $1$. To this end:
    \begin{align*}
        \frac{\left(\frac{1}{2} - \varepsilon\right)^{\frac{1}{p-1}}}{\left(\frac{1}{2} - \delta\right)^{\frac{p}{p-1}} + \left(\frac{1}{2} - \varepsilon\right)^{\frac{p}{p-1}}} \tag{$\delta > \varepsilon$} &> \frac{\left(\frac{1}{2} - \varepsilon\right)^{\frac{1}{p-1}}}{2\left(\frac{1}{2} - \varepsilon\right)^{\frac{p}{p-1}}}\\ &= \frac{1}{1-2\varepsilon} > 1.
    \end{align*}
    The lemma stands proved.  
\end{proof}

\end{document}